\documentclass[11pt]{article}

\usepackage[letterpaper,top=1in,bottom=1in,left=1in,right=1in]{geometry}

\usepackage{graphicx}
\usepackage{tikz}
\usetikzlibrary{arrows.meta, positioning, shapes.geometric}
\usepackage{booktabs}
\usepackage{lineno}
\usepackage{enumerate}
\usepackage[shortlabels]{enumitem}
\usepackage{algorithm}
\usepackage[noend]{algpseudocode}
\usepackage{appendix}
\usepackage{todonotes}
\usepackage{xcolor}
\usepackage{wrapfig}

\usepackage{amsmath, amsfonts}
\usepackage{amsthm}
\usepackage{mathtools}
\usepackage{thmtools}
\usepackage{thm-restate}

\usepackage[
	style=alphabetic,
	backref=true,
	doi=true,
	url=false,
	maxcitenames=3,
	mincitenames=3,
	maxbibnames=10,
	minbibnames=10,
    maxalphanames=6,
    minalphanames=3,
	backend=biber,
	sortlocale=en_US
]{biblatex}
\AtEveryCitekey{\clearfield{note}}

\usepackage{hyperref}
\usepackage[capitalize]{cleveref}

\makeatletter
\crefname{ALG@line}{line}{lines}
\Crefname{ALG@line}{Line}{Lines}

\AddToHook{env/algorithmic/begin}{\def\@currentcounter{ALG@line}}
\makeatother

\declaretheorem[name=Theorem]{theorem}
\declaretheorem[name=Claim, sibling=theorem]{claim}
\declaretheorem[name=Lemma, sibling=theorem]{lemma}

\declaretheorem[name=Definition, sibling=theorem]{definition}

\declaretheorem[name=Conjecture, sibling=theorem]{conjecture}
\graphicspath{ {figures/} }

\newcommand{\pa}[1]{\left(#1\right)}

\newcommand{\set}[1]{\left\{#1\right\}}
\newcommand{\abs}[1]{\left|#1\right|}

\newcommand{\cO}{O}

\newcommand{\cA}{\mathcal{A}}
\newcommand{\tO}{\widetilde{O}}

\newcommand{\ol}{\overline}

\newcommand{\reversed}[1]{\overleftarrow{#1}}

\DeclareMathOperator{\poly}{poly}
\DeclareMathOperator{\volop}{vol}

\newcommand{\inc}[2][G]{\partial_{#1}\!\pa{#2}}
\newcommand{\incout}[2][G]{\partial_{#1}^{+}\!\pa{#2}}
\newcommand{\incin}[2][G]{\partial_{#1}^{-}\!\pa{#2}}

\newcommand{\dg}[2][G]{\deg_{#1}\!\pa{#2}}
\newcommand{\dgout}[2][G]{\deg_{#1}^{+}\!\pa{#2}}
\newcommand{\dgin}[2][G]{\deg_{#1}^{-}\!\pa{#2}}

\newcommand{\mindeg}[1][G]{\delta\!\pa{#1}}
\newcommand{\mindegout}[1][G]{\delta^{+}\!\pa{#1}}
\newcommand{\mindegin}[1][G]{\delta^{-}\!\pa{#1}}
\newcommand{\maxdeg}[1][G]{\Delta\!\pa{#1}}
\newcommand{\maxdegout}[1][G]{\Delta^{+}\!\pa{#1}}
\newcommand{\maxdegin}[1][G]{\Delta^{-}\!\pa{#1}}

\newcommand{\nbr}[2][G]{N_{#1}\!\pa{#2}}
\newcommand{\nbrout}[2][G]{N_{#1}^{+}\!\pa{#2}}
\newcommand{\nbrin}[2][G]{N_{#1}^{-}\!\pa{#2}}

\newcommand{\bnd}[2][G]{\partial_{#1}\!\pa{#2}}
\newcommand{\bndout}[2][G]{\partial_{#1}^{+}\!\pa{#2}}
\newcommand{\bndin}[2][G]{\partial_{#1}^{-}\!\pa{#2}}

\newcommand{\vol}[2][G]{\volop_{#1}\!\pa{#2}}
\newcommand{\volout}[2][G]{\volop_{#1}^{+}\!\pa{#2}}
\newcommand{\volin}[2][G]{\volop_{#1}^{-}\!\pa{#2}}

\newcommand{\cutsize}{\kappa}

\newcommand{\volest}{\nu}
\newcommand{\event}{\mathcal E}

\newcommand{\volcut}[3]{$(#1, #2\!#3)$-separation triple}
\newcommand{\volvertcut}[4]{$#4$-rooted $(#1, #2\!#3)$-separation triple}
\newcommand{\voledgecut}[4]{$#4$-rooted $(#1, #2\!#3)$-edge cut}

\newcommand{\ready}[2]{$(#1, #2)$-ready}

\newcommand{\ceil}[1]{\lceil #1 \rceil}
\newcommand{\floor}[1]{\lfloor #1 \rfloor}
\newcommand{\boosts}{\texttt{boosts}}

\newcommand{\mh}[2][says]{\textcolor{red}{\texttt{MH #1:}} {\color{red} #2}}

\newcommand{\runtimead}{{$O(\kappa^4 \log^2(n))$}}
\newcommand{\runtimeob}{{$O(\kappa^3 \log^2(n))$}}
\title{Incremental Exact Vertex Connectivity}
\author{
Felix Adena\thanks{Institute of Science and Technology Austria (ISTA), Klosterneuburg, Austria}
\and
Christoph Kern\footnotemark[1]
\and
Monika Henzinger\footnotemark[1]
\and
Daniel Rutschmann\footnotemark[1]
}

\date{}

\begin{document}

\maketitle

\begin{abstract}
We study the problem of maintaining the vertex connectivity $\kappa$ of a graph under edge insertions.
For undirected graphs, our algorithm uses $O(\kappa^4 \log^2 n)$ amortized time per edge insertion.
After each edge insertion, the algorithm outputs the current value of $\kappa$, plus a set of $\kappa$ vertices whose removal disconnects the graph.
Our algorithm is Monte-Carlo randomized and works with a randomness-adaptive adversary.
The algorithm also works on directed graphs against an oblivious adversary, which is equivalent to an output-adaptive adversary if we only have to output the current value of $\kappa$.

Previously, algorithms using SPQR trees could maintain the vertex connectivity in undirected graphs
with $O(\log n)$ time per edge insertion for $\kappa \le 3$,
but no algorithm with $o(n)$ time per edge insertion was known for $\kappa \ge 4$.
Moreover, for $\kappa \ge 4$, on general directed graphs or sparse undirected graphs, no algorithm faster than recomputation from scratch was known.

We complement our algorithm with a conditional lower bound: Under the OMv conjecture, any insertions-only algorithm that outputs the current value of $\kappa$ after every insertion needs $\Omega(n^{1-\varepsilon})$ amortized time per edge insertion for $\kappa = \Theta(n)$.
Previous conditional lower bounds for this problem only apply to worst case running time, or algorithms that support both edge insertions and edge deletions.

\end{abstract}

\subparagraph{Funding} This research was supported by the European Research Council (ERC) under the European Union's Horizon 2020 research and innovation programme (Grant agreement No.\ 101019564), and the Austrian Science Fund (FWF) under grant DOI 10.55776/Z422.

For open access purposes, the authors have applied a CC BY public copyright license to any author-accepted manuscript version arising from this submission.
Views and opinions expressed are however those of the author(s)
only and do not necessarily reflect those of the European Union or the 
\begin{wrapfigure}[3]{r}[0pt]{4.7cm}%
\vspace{-2em}
\includegraphics[width=5cm]{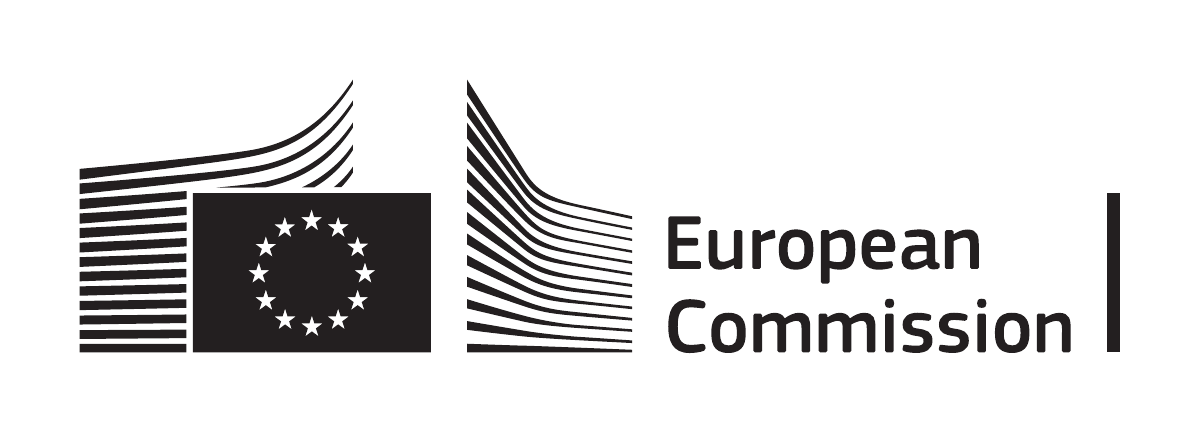}
\end{wrapfigure}
European Research Council Executive Agency.
Neither the European Union nor the granting authority can be held responsible for them.
\vspace{1em}

\pagenumbering{gobble}
\pagebreak
\pagenumbering{arabic}

\section{Introduction}
The \textit{vertex connectivity} $\kappa(G)$ of a (directed) graph $G=(V,E)$ represents the minimum number of \textit{vertices} whose removal causes $G$ to no longer be strongly connected \textit{or} turns the graph into a 1-vertex graph, whereas the \textit{edge connectivity} $\lambda(G)$ denotes the minimum number of \textit{edges} whose removal causes $G$ to no longer be strongly connected. 
Computing the vertex and edge connectivity are fundamental problems needed, for example, to determine the resilience of a network to vertex and edge failures, respectively.

\textbf{Edge connectivity.} In the past, a large body of work has concentrated on computing \textit{edge connectivity in undirected graphs}, leading to $\tilde O(m)$\footnote{The $\tilde O(\cdot)$ notation hides polylogarithmic factors.} time randomized~\cite{karger1996new} and deterministic~\cite{kawarabayashi2015deterministic,DBLP:journals/siamcomp/HenzingerRW20,DBLP:conf/soda/HenzingerLRW24} algorithms, where $m=|E|$. The problem has also been studied in the \textit{dynamic setting}, where each update operation can insert or delete one edge:
%in the unweighted setting 
It can be solved in update time sublinear in $m$~\cite{DBLP:conf/soda/GoranciHNSTW23,DBLP:journals/algorithmica/VosC26} and in polylogarithmic amortized time per insertion in the insertions-only setting~\cite{DBLP:journals/talg/GoranciHT18}\footnote{There is also a large body of work on $(1+\varepsilon)$-approximate dynamic edge connectivity algorithms, omitted here to keep the presentation simple.}.
In \textit{directed graphs}, our understanding still lags behind: Currently there are three different algorithms
with incomparable running times, namely $O(\min\{\frac{n}{m^{1/3}}, \sqrt{n}\} \cdot m^{1+o(1)})$~\cite{DBLP:conf/focs/Cen0NPSQ21}, $O(\kappa m)$~\cite{DBLP:journals/jcss/Gabow95}, and $O(n^2)$~\cite{DBLP:conf/icalp/ChekuriQ21}, where $n = |V|$. In the fully dynamic setting\footnote{The fully dynamic setting allows both edge insertions and deletions.} there are strong conditional lower bounds showing that even just testing whether the graph is strongly connected (which is equivalent to testing whether $\lambda >0$) takes near-linear time per update or query (based on SETH)~\cite{DBLP:conf/focs/AbboudW14} or $\Omega(m^{1/2-\varepsilon})$  (based on OMv)
~\cite{DBLP:conf/stoc/HenzingerKNS15}. In the fully dynamic setting these lower bounds are for the \textit{amortized} running time; in the insertions-only or deletions-only setting these bounds are for the \textit{worst-case} running time.
Recently,
insertions-only and deletions-only algorithms for testing strong connectivity and 2-edge connectivity in $O(m^{1+o(1)})$ and $\tilde O(n^2)$ total time for all updates, respectively, were given~\cite{chen2024almost,van2024almost,DBLP:conf/esa/GeorgiadisGI25}. There is no known insertions-only or deletions-only algorithm for computing edge-connectivity in directed graphs better than recomputation from scratch.

%deterministic algorithms for the insertions-only~\cite{chen2024almost} and deletions-only~\cite{van2024almost} setting where developed with $O(m^{1+o(1)})$ total time for all operations

\textbf{Vertex connectivity.} As $k$-vertex connectivity implies $k$-edge connectivity but not vice versa, computing the vertex connectivity is generally considered more challenging than computing the edge connectivity.
Motivated by the progress for edge connectivity, computing \textit{vertex connectivity} has also seen renewed interest: 
In a recent series of results, the running time for \textit{undirected graphs} was reduced to $\tilde{O}(m + n \kappa^3)$~\cite{DBLP:conf/stoc/NanongkaiSY19,DBLP:conf/soda/ForsterNYSY20} and to $O(m^{1+o(1)})$~\cite{DBLP:conf/stoc/LiNPSY21}.
In \textit{directed graphs} the $\tilde O(nm)$ running time of the randomized algorithm of~\cite{DBLP:journals/jal/HenzingerRG00} was recently improved in four ways: (1) by a deterministic algorithm~\cite{DBLP:conf/stoc/JiangNSY25} with almost the same running time, namely $O(mn^{1+o(1)})$, (2) by randomized algorithms with $O(\min\{m n^{11/12 + o(1)}, n^{2+o(1)}\})$~\cite{DBLP:conf/stoc/LiNPSY21} and $O(\min\{m^{1+o(1)} \kappa, n^{2+o(1)}\})$~\cite{DBLP:conf/soda/ChuzhoyMT26} running time, (3) by a reduction~\cite{DBLP:conf/sosa/FischerJMY26} to the undirected setting, resulting in a randomized algorithm with $O(n^{2+o(1)})$ running time, and (4) by an algorithm with time linear in $m$ \textit{and} polynomial in $\kappa$, namely $\tilde{O}(m \kappa^2)$, by \cite{DBLP:conf/soda/ForsterNYSY20}.

However, the \textit{dynamic setting} for vertex connectivity is still barely understood.
In \emph{undirected graphs} the classic results consist of (1) a $(2+\varepsilon)$-approximate insertions-only algorithm with $O(n/\varepsilon)$ amortized time per update~\cite{DBLP:journals/jal/Henzinger97}, (2) algorithms for $k \le 4$, namely
(2a) fully dynamic algorithms for testing whether, for a fixed value $k = 2, 3$ or $4$, the graph is $k$-connected and, if not, returning a cut of size at most $k-1$, in $O(\poly \log n)$~\cite{holmLT2001,DBLP:conf/stoc/HolmNR025}\footnote{This is not spelled out in their papers, but their top trees should contain sufficient information to do this.}; $O(n)$, and $O(n\alpha(n,n))$\footnote{Here $\alpha$ denotes the two-parameter inverse Ackermann function}\cite{DBLP:journals/jacm/EppsteinGIN97} time per update, respectively, and
(2b) insertions-only algorithms in the same setting for $k = 2$ with $O(\log n)$ worst-case update time~\cite{PoutreW98},
and $k=3$ with $O(\log n)$ amortized update time~\cite{DBLP:journals/algorithmica/BattistaT96} or $O(\log^3 n)$ worst-case update time~\cite{DBLP:conf/soda/HolmR20},
and (3) testing whether the graph is $k$-connected for a fixed value $k$ in time $O(n^{1.575} + k^2n)$ per update~\cite{DBLP:conf/soda/Sankowski07}.
Note that for sparse undirected graphs the results for $k=4$ are almost the same as the best static algorithm.

Additionally, given
two vertices $s$ and $t$ as query parameters, there are (4) fully dynamic algorithms that can test whether $s$ and $t$ are 2-connected in $O(\log n)$ time and process updates in $\tilde O(\log^2 n)$ time~\cite{DBLP:conf/stoc/HolmNR025},
and (5) incremental algorithms based on SPQR trees that can test whether $s$ and $t$ are 3-connected and process updates in $O(\log n)$ amortized time~\cite{DBLP:journals/algorithmica/BattistaT96} or $O(\log^3 n)$ worst-case time~\cite{DBLP:conf/soda/HolmR20}.
Comparing with the performance for dynamic edge connectivity in undirected graphs, it is interesting to ask (A) whether there exists an insertions-only algorithm with polylogarithmic amortized update time, and (B) whether there exists a sublinear (in $m$) time algorithm for the fully dynamic setting. Furthermore, this leads to question (C) whether an 
 insertions-only or deletions-only algorithm with amortized running time $\tilde O(\poly(\kappa))$ is possible.

In \textit{directed graphs} the same conditional  lower bounds for fully dynamic algorithms as above for edge connectivity apply, showing that even  testing whether $\kappa >0$ takes at least near-linear time per update or query (from SETH)~\cite{DBLP:conf/focs/AbboudW14} or $\Omega(m^{1/2-\varepsilon})$  (from OMv)
~\cite{DBLP:conf/stoc/HenzingerKNS15}. As mentioned above, these lower bounds also apply to incremental/decremental algorithms with worst-case (but not amortized) running time.
This leads to the question (D) whether an 
 insertions-only or deletions-only algorithm with amortized running time $\tilde O(\poly(\kappa))$ is also possible in directed graphs.

%there are strong conditional lower bounds: Based on SETH any algorithm with polynomial preprocessing time that decides whether $\kappa >0$ (i.e.~testing strong connectivity) either has update time $\Omega(m^{1-\varepsilon})$ or query time $\Omega(m^{1-\varepsilon})$~\cite{DBLP:conf/focs/AbboudW14} and based on the OMv conjecture any algorithm with polynomial preprocessing time that decides whether $\kappa >0$ either has update time $\Omega(m^{1/2-\varepsilon})$ or query time $\Omega(m^{1-\varepsilon})$
The goal of this paper is to answer these questions. More specifically, we study algorithms for dynamic vertex connectivity both in directed and undirected graphs.

In our analysis we distinguish between two different notions of modeling an adaptive adversary, similarly to \cite{DBLP:conf/stoc/NanongkaiS17, DBLP:conf/stoc/BrandLS23, DBLP:conf/soda/0003ST25, DBLP:conf/nips/BanihashemGGHJM25}.
An \emph{output-adaptive} adversary gets to pick the next update or query based on the algorithm's previous outputs, but not its internal state or random bits;
and a \emph{randomness-adaptive} adversary additionally has full knowledge of the algorithm's internal state, including past random choices, but not future random bits.
In contrast, an \emph{oblivious} adversary has to fix the full update sequence in advance, i.e., independently of the algorithm's outputs and internal state.
We remark that, for problems where each query has a unique correct answer,
any algorithm that works with an oblivious adversary automatically also works with an output-adaptive adversary.

% \iffalse
\begin{table}[htbp]
    \centering
    \renewcommand{\arraystretch}{1.4}
    \begin{tabular}{@{} l p{0.25\textwidth} p{0.55\textwidth} @{}}
        \toprule
        \textbf{Graph} & \textbf{Edge Connectivity} & \textbf{Vertex Connectivity} \tabularnewline
        \midrule
        \textbf{Undirected} 
        & \raggedright $\poly\log(n)$~\cite{DBLP:journals/talg/GoranciHT18} 
        & \raggedright \textbf{$\tO(\poly(\kappa))$ with randomness-adaptive adversary} \tabularnewline
        \addlinespace
        & 
        & \raggedright \textbf{Amortized $\mathbf{\Omega(\kappa^{3-\varepsilon} / m)}$ update} \newline
          \textbf{or $\mathbf{\Omega(\kappa^{2-\varepsilon})}$ query time} \tabularnewline
        \midrule
        \textbf{Directed} 
        & --- 
        & \raggedright \textbf{$\tO(\poly(\kappa))$ with oblivious adversary} \tabularnewline  
        \addlinespace
        & 
        & \raggedright \textbf{Amortized $\mathbf{\Omega(\kappa^{3-\varepsilon} / m)}$ update} \newline
          \textbf{or $\mathbf{\Omega(\kappa^{2-\varepsilon})}$ query time} \tabularnewline
        \addlinespace
        & 
        & \raggedright Worst-case $\Omega(m^{1-\varepsilon})$ update \newline
          or $\Omega(m^{2-\varepsilon})$ query time~\cite{DBLP:conf/focs/AbboudW14,DBLP:conf/stoc/HenzingerKNS15} \tabularnewline
        \bottomrule
    \end{tabular}
    \caption{Running time for insertions-only algorithms; our new results are in bold. The results hold for arbitrary $\kappa$, including $\kappa = \Theta(n)$. All lower bounds are conditional lower bounds.}
    \label{tab:connectivity_metrics}
\end{table}
% \fi
\iffalse
\begin{table}[htbp]
\centering
\label{tab:connectivity_metrics}
\renewcommand{\arraystretch}{1.3}
\begin{tabular}{@{}lll@{}}
\toprule
\textbf{Graph} & \textbf{Edge Connectivity} & \textbf{Vertex Connectivity} \\
\midrule
Undirected       & $\poly\log(n)$~\cite{DBLP:journals/talg/GoranciHT18}           & \textbf{$\tilde O(\poly(\kappa))$ with randomness-adaptive adversary}               \\
       &            & \textbf{amortized $\mathbf{\Omega(\kappa^{3-\varepsilon} / m)}$ update or $\mathbf{\Omega(\kappa^{2-\varepsilon})}$ query time }               \\
\hline
Directed       &  &\textbf{$\tilde O(\poly(\kappa))$ with oblivious adversary}          \\  

      &         & \textbf{amortized $\mathbf{\Omega(\kappa^{3-\varepsilon} / m)}$  update or $\mathbf{\Omega(\kappa^{2-\varepsilon})}$ query time }    \\
     &        & worst-case $\Omega(m^{1-\varepsilon})$  update or $\Omega(m^{2-\varepsilon})$ query time~\cite{DBLP:conf/focs/AbboudW14,DBLP:conf/stoc/HenzingerKNS15}              \\
\bottomrule
\end{tabular}
\caption{Running time for insertions-only algorithms, our new results are in bold. The results hold for arbitrary $\kappa$, including $\kappa = \Theta(n)$. All lower bounds are conditional lower bounds. \mh{feel free to make the table look better}}
\end{table} 
\fi
\smallskip
\textbf{Our contribution.}
In this paper we show the following two results (see also \cref{tab:connectivity_metrics}):

(1) We give an insertions-only algorithm for \textit{vertex connectivity in undirected graphs} with amortized update time \runtimead, where $\kappa$ is the vertex connectivity in the final graph.
It works with a randomness-adaptive adversary and can either return $\kappa$ in constant time or a minimum vertex cut in time linear in its size.
Our algorithm also works for \textit{vertex connectivity in directed graphs} if the adversary is oblivious, or, if the adversary is output-adaptive and the algorithm only has to return the value $\kappa$.
For these types of adversaries, the running time improves to \runtimeob, both on directed and undirected graphs.
Note that the above mentioned lower bounds that are \textit{polynomial in $n$} do not apply as our algorithm achieves a bound of $\tilde O(\poly(\kappa))$ on the \textit{amortized} time per insertion. Our algorithm answers questions (C) and (D) positively and
is the first dynamic algorithm to achieve such a bound.

(2) We give a novel conditional lower bound based on the OMv assumption: There is no insertions-only (or deletions-only) algorithm for minimum vertex cuts on unweighted, undirected graphs that, given polynomial (in $n$ and $m$) preprocessing time, supports edge insertions (or deletions) in \textit{amortized} time $O(\kappa^{3-\varepsilon} / m)$ and queries (that return the value of $\kappa$) in $O(\kappa^{2-\varepsilon})$. There are no assumptions on $\kappa$. For $\kappa \in \Theta(n)$, these bounds are $O(n^{1-\varepsilon})$ for insertions (or deletions), and $O(n^{2-\varepsilon})$ for queries.

Note that this lower bound answers question (A) in the negative and (for graphs with  $\kappa = \Theta(n)$) also answers question (B) in the negative. Said differently, it helps to explain the lack of progress in finding insertions-only algorithms with running time polylogarithmic in $n$ (and \textit{not} polynomial in $\kappa$): No such algorithm can exist. It also shows a dichotomy between computing edge connectivity and vertex connectivity and substantiates the intuition that vertex connectivity is more challenging: As mentioned above, there exists an insertions-only algorithm for edge connectivity with polylogarithmic in $n$ (and independent of the edge connectivity value) running time, while no such algorithm can exist for computing vertex connectivity.

The remainder of the paper is organized as follows:
We introduce all necessary definitions and asymptotic notation in \cref{sec:prelim}, followed by a comprehensive technical overview of our insertions-only algorithm and its distinct phases in \cref{sec:overview}.
In \cref{sec:lower_bound}, we formally establish our conditional lower bounds based on the OMv conjecture.
\Cref{sec:analysis} contains the rigorous running time and correctness analysis of our algorithm.
The subsequent sections are dedicated to the proofs and analysis of the subroutines used by the algorithm: \Cref{sec:sampling} details our randomized sampling procedures, \cref{sec:dyn_rooted_vc} presents the incremental data structure that is the core of our algorithm, and \cref{sec:aux_datastructs} concludes with some auxiliary data structures from the literature that we use to handle special cases, such as the graph not being strongly connected.

\section{Preliminaries}\label{sec:prelim}

Throughout this paper, all logarithms are base-2 unless otherwise specified.
\emph{With high probability} (w.h.p.) means with probability at least $1 - n^{-c}$ for some constant $c \ge 1$.
Unless explicitly stated otherwise, we always consider \emph{simple} graphs, that is, graphs without self-loops or parallel edges.

\paragraph{Graphs and Neighborhoods.}
Let $G = (V, E)$ be a directed or undirected graph with $n = \abs{V}$ vertices and $m = \abs{E}$ edges.
For a subset of vertices $S \subseteq V$, we let $G \setminus S$ denote the subgraph induced by $V \setminus S$.
For two disjoint vertex sets $S, T \subseteq V$, we define $E_G(S, T) = \set{(u, v) \in E \mid u \in S, v \in T}$ as the set of directed edges from $S$ to $T$.
In undirected graphs, edges are unordered, and $E_G(S, T)$ denotes the set of edges with one endpoint in $S$ and the other in $T$.
We drop the subscript $G$ when the graph is clear from the context.

For an edge $e = (u, v)$, denote by $\reversed{e} = (v, u)$ the \emph{reversed edge}.
For a graph $G = (V, E)$, we denote by $\reversed{G} = (V, \{\reversed{e}: e \in E\})$ the \emph{reversed graph}, with the same vertex set but all edges reversed.

For a vertex $v \in V$, we denote its set of \emph{outgoing} and \emph{incoming edges} as $\incout{v} = E_G(\set{v}, V \setminus \set{v})$ and $\incin{v} = E_G(V \setminus \set{v}, \set{v})$, respectively.
The total set of \emph{incident edges} is $\inc{v} = \incout{v} \cup \incin{v}$.
The corresponding \emph{degrees} are $\dgout{v} = \abs{\incout{v}}$, $\dgin{v} = \abs{\incin{v}}$, and $\dg{v} = \dgout{v} + \dgin{v}$.
The \emph{out-neighbors} and \emph{in-neighbors} of $v$ are $\nbrout{v} = \set{u \mid (v, u) \in E}$ and $\nbrin{v} = \set{u \mid (u, v) \in E}$.
The \emph{minimum} and \emph{maximum out-degree} of $G$ are $\mindegout = \min_{v \in V} \dgout{v}$ and $\maxdegout = \max_{v \in V} \dgout{v}$, and analogously $\mindegin, \maxdegin$ for in-degrees and $\mindeg, \maxdeg$ for the total degree.
In undirected graphs, we drop the orientation superscripts.

\paragraph{Volumes and Boundaries.}
For a vertex subset $A \subseteq V$, we define its \emph{boundary edges} as the edges leaving $A$, denoted by $\bndout{A} = E_G(A, V \setminus A)$, and the edges entering $A$, denoted by $\bndin{A} = E_G(V \setminus A, A)$.
The total boundary is $\bnd{A} = \bndout{A} \cup \bndin{A}$.
This uses the same notation as the incident edges of a vertex introduced above, where $\inc{v} = \bnd{\set{v}}$.
Similarly, we overload the neighborhood notation to a set of vertices, defining its out- and in-neighbors as the endpoints of the boundary edges outside $A$, that is $\nbrout{A} = \set{v \mid (u, v) \in \bndout{A}}$ and $\nbrin{A} = \set{u \mid (u, v) \in \bndin{A}}$, with $\nbr{A} = \nbrout{A} \cup \nbrin{A}$.
We define the \emph{out-volume} and \emph{in-volume} of the set $A$ as the sum of the respective degrees of its vertices, meaning $\volout{A} = \sum_{v \in A} \dgout{v}$ and $\volin{A} = \sum_{v \in A} \dgin{v}$.
The total volume is $\vol{A} = \volout{A} + \volin{A}$.

\paragraph{Vertex Connectivity and Cuts.}
A directed graph $G$ is \emph{strongly connected} if there is a directed path between every pair of vertices.
A \emph{vertex cut} is a set $S \subseteq V(G)$ such that $G \setminus S$ is not strongly connected or consists of a single vertex.
The \emph{size} of the vertex cut is $\abs{S}$.
An ordered partition of the vertex set $V$ into three disjoint sets $(L, S, R)$ such that $L$ and $R$ are non-empty and there are no edges from $L$ to $R$,
meaning $E_G(L, R) = \emptyset$, is called a \emph{separation triple}.
If $(L, S, R)$ is a separation triple, then $S$ is always a vertex cut.
Conversely, every vertex cut $S$ with $|S| \ne n-1$ lies in some separation triple $(L, S, R)$:
for example, we may take $L$ to be a strongly connected component of $G \setminus S$ with no outgoing edges.

If $\abs{S} = \cutsize$, we refer to $(L, S, R)$ as a $\cutsize$-separation triple.
For two specific vertices $s$ and $t$, an \emph{$(s, t)$-separation triple} is a separation triple $(L, S, R)$ where $s \in L$ and $t \in R$.
The \emph{vertex connectivity} of $G$, denoted by $\cutsize(G)$, is the minimum size of a vertex cut in $G$.
A graph is \emph{$k$-connected} if $\cutsize(G) \ge k$.

\paragraph{Dynamic Setting.}
In the \emph{incremental} dynamic setting, the graph starts empty (or with an initial edge set) and undergoes a sequence of edge insertions.
We generally let $m$ denote the total number of edges in the graph after all insertions have been processed.

% BUILD_COMMAND make main.pdf -B
\section{Overview}\label{sec:overview}

We give an overview of our techniques.
As we use different subroutines depending on whether one side of the cut has small volume or not, we need the following definition:
\begin{definition}
    Let $G$ be a graph. A \emph{\volcut{\cutsize}{\le}{\volest}} is a separation triple $(L, S, R)$ with $|S| = \cutsize$ and $\volout{L} \le \volest$. 
    We similarly define a \emph{\volcut{\cutsize}{<}{\volest}}, \emph{\volcut{\cutsize}{\ge}{\volest}}, etc.
    For a vertex $s \in V(G)$, an \emph{\volvertcut{\cutsize}{\le}{\volest}{s}} is a \volcut{\cutsize}{\le}{\volest} with $s \in L$.
    We say that the graph $G$ is \emph{$(\kappa, \nu)$-ready} if $G$ is $\kappa$-connected and has no \emph{\volcut{\cutsize}{<}{\volest}}.
\end{definition}

For ease of presentation, suppose that the initial graph $G$ is strongly connected.
Our algorithm operates in rounds.
At the start of each round, we assume that $G$ is $\cutsize$-connected,
and during the round, we process edge insertions until $G$ has no vertex cuts of size $\cutsize$.
Then, we increment $\cutsize$ and start the next round.
We may assume that $\cutsize$ stays below $\frac{n}{2}$:
Indeed, once $\cutsize \ge \frac{n}{2}$, we can afford to run the algorithm of \cite{DBLP:conf/soda/ChuzhoyMT26,DBLP:conf/sosa/FischerJMY26} after every edge insertion, since $O(n^{2+o(1)}) \subseteq O(\kappa^3)$.

Each round consists of $O(\log m)$ phases, indexed by a parameter $\volest = 1, 2, 4, 8, \ldots, m$.
At the start of each phase, we require both $G$ and the reversed graph $\reversed{G}$ to be $(\kappa, \nu)$-ready.
The phase first processes edge insertions until there are no \volcut{\cutsize}{<}{2 \nu}s in $G$,
at which point $G$ is $(\kappa, 2 \nu)$-ready.
We do the same with $\reversed{G}$.
Thus, at the end of the phase, both $G$ and $\reversed{G}$ are $(\kappa, 2 \nu)$-ready.
Hence, every remaining $\cutsize$-separation triple $(L, S, R)$ has $\volout{L} \ge 2 \nu$ and $\volin{R} \ge 2 \nu$.

To implement a phase, we first sample a set $T$ of terminals
    with the goal that $L \cap T \neq \emptyset$ for any $\kappa$-separation triple $(L, S, R)$ with large enough volume $\volout{L} \geq \volest$.
For this, we adapt the sampling procedure of \cite{Mosenzon26} that samples each vertex $u$ with probability roughly $2 \dgout{u} / \nu$.

\begin{restatable}[\textsc{SampleRoots}, modified from {\cite[Lemma 2.5]{Mosenzon26}}]{theorem}{thmSampling} \label{thm:sampling}
    There is a randomized algorithm that, given a directed graph $G$ on $n$ vertices and $m$ edges and a volume parameter $\volest$, outputs a set $T \subseteq V(G)$ of size at most $\frac{2 m}{\volest}$, such that for every set $A \subseteq V(G)$ with $\volout{A} \geq \volest$, we have $\Pr[A \cap T = \emptyset] \leq \frac{1}{2}$.
    The algorithm runs in $O(m)$ time.
\end{restatable}

Then, for each vertex $s \in T$, we build an incremental data structure that reports an $s$-rooted vertex cut of size $\cutsize$ whenever an \volvertcut{\cutsize}{\le}{2\volest}{s} exists, and otherwise certifies that no such triple exists.
Since we instantiate one such structure per vertex $s \in T$, its initialization must run in $O(\volest \poly(\cutsize) \poly \log n)$ time rather than $\Omega(m)$.
\cite{DBLP:conf/soda/ForsterNYSY20} give a randomized subroutine $\textsc{LocalVC}(G, s, \volest, \cutsize)$ that solves this problem in the static setting in $O(\volest \cutsize^2)$ time.
This subroutine grows a search tree from $s$, so we can make it incremental by maintaining this tree under edge insertions via single-source incremental reachability, similar to how \cite{GuptaK21} adapt the Ford--Fulkerson max-flow algorithm to the incremental setting; we pause this growth process whenever it finds a cut and resume it as new edges arrive.
We thus obtain:

\begin{restatable}[\textsc{IncrementalRootedVC}]{theorem}{theoLocalVC} \label{thm:local_vc}
    Given a graph $G$ with $m$ edges and two positive integers $\cutsize, \volest$ with $\kappa \le \frac{n}{2}$, $\volest \leq \frac{m}{24 \cutsize}$, and $\mindegout[G] > \kappa$, as well as a vertex $s \in V(G)$,
    there is a data structure that maintains $G$ under $\ell$ edge insertions supplied by a randomness-adaptive adversary $\cA$.
    After every insertion, the data structure reports either the vertex cut $S$ of some \volvertcut{\cutsize}{\le}{6 \volest \cutsize}{s} or the symbol $\bot$.
    Let $(L, S, R)$ be one particular \volvertcut{\cutsize}{\le}{\volest}{s} in $G$.
    With probability $\ge \frac{1}{2}$, the data structure only outputs $\bot$ once $(L, S, R)$ is no longer a valid \volvertcut{\cutsize}{\le}{\volest}{s}.
    The total running time is $\cO(\volest  \cutsize^2 + \ell)$.
\end{restatable}

\cref{thm:local_vc} requires $\mindegout[G] > \kappa$, that is, that the graph contains no trivial vertex cuts.
We handle these trivial cuts by maintaining the in- and out-degree of each vertex, plus a linked list of all vertices with in- or out-degree equal to $\cutsize$ at the start of each round.
We do this until there are no trivial vertex cuts.
Updating the degrees or the linked list takes constant time per edge insertion, so this additional step takes $O(n)$ time per round, or $O(\cutsize n)$ total time across all rounds.

\cref{thm:sampling,thm:local_vc} are Monte-Carlo randomized and only have constant success probability to detect any particular \volvertcut{\cutsize}{<}{2 \volest}{s}.
If the adversary is oblivious, then we can consider the last remaining \volvertcut{\cutsize}{<}{2 \volest}{s} $(L, S, R)$
right before the moment when no such triple exists, and show that our algorithm detects this separation triple.
Thus, it suffices to perform $O(\log n)$ repetitions to boost the success probability of \cref{thm:sampling,thm:local_vc} to $1 - n^{-5}$.
This works both on directed and undirected graphs.

If the adversary is randomness-adaptive, then we cannot argue about a last remaining \volvertcut{\cutsize}{<}{2 \volest}{s}.
For undirected graphs, we show an upper bound of $O(n^{\kappa+1})$ for the number of \volvertcut{\cutsize}{<}{2 \volest}{s}s. Thus, it suffices to perform $O(\kappa \log n)$ repetitions to boost the success probability to $1 - n^{-\Omega(\kappa)}$, and union bound over all such triples. This approach does not work on directed graphs: while there are $O(n^{\kappa})$ choices for $S$, there might be exponentially many choices for $L$. This highlights a subtle difference between vertex cuts and separation triples.

\begin{restatable}{lemma}{lemmTripleCount} \label{lemm:triple_count}
    If $G$ is an \emph{undirected} graph that is $(\kappa, \nu)$-ready, then there are at most $n^{\kappa+1}$ triples $(L, S, R)$ that are \volcut{\cutsize}{<}{2 \nu}s.
\end{restatable}

There is one caveat: \cref{thm:local_vc} only works when the volume is not too large, namely $\nu \leq \frac{m}{24 \cutsize}$.
However, if $\volest \ge \frac{m}{24 \cutsize}$, then we are guaranteed by \cref{thm:sampling} that the size of $T$ will be small, i.e., $|T| \in O(\cutsize)$.
    Recall that $T$ was sampled via \cref{thm:sampling} with the goal that $L \cap T \ne \emptyset$ for any $\kappa$-separation triple $(L, S, R)$ with $\volout{L} \ge \nu$.
    We also apply \cref{thm:sampling} to the reversed graph $\reversed{G}$ to get a set $\reversed{T}$ with the goal that $R \cap \reversed{T} \ne \emptyset$ for any $\kappa$-separation triple $(L, S, R)$ with $\volin{R} \ge \nu$.
Then, we could afford to use an incremental exact max-flow data structure to maintain an $(s, t)$-vertex cut of size $\cutsize$ for all pairs $(s, t) \in T \times \reversed{T}$.
There are $|T|\cdot |\reversed{T}| = \Theta(\cutsize^2)$ such pairs $(s, t)$, so we would need $\Theta(\cutsize^2)$ incremental exact max-flow data structures.
We can do slightly better and reduce the number of pairs to $\Theta(\cutsize)$, via a different sampling approach that just samples $O(1)$ vertices uniformly at random with repetition.

\begin{restatable}[\textsc{UnbiasedSample}]{lemma}{lemmUnbiased} \label{lemm:unbiased_sampling}
    Let $G$ be a directed graph with $n$ vertices and $m$ edges.
    Let $u$ be a single vertex sampled uniformly at random.
    Then, for any separation triple $(L, S, R)$ with $|S| = \cutsize \le \frac{n}{2}$, with probability at least $\frac{1}{2}$, we have $u \in (L \cup R)$.
\end{restatable}

In addition to $T$ and $\reversed{T}$, we sample a vertex $u$ via \cref{lemm:unbiased_sampling}.
Let $(L, S, R)$ be any $\kappa$-separation triple with $|S| = \kappa \le \frac{n}{2}$, $\volout{L} \ge \nu$, and $\volin{R} \ge \nu$.
With constant probability, the vertex $u$ lies in $L \cup R$. 
If $u$ lies in $L$, then by \cref{thm:sampling}, with constant probability, $\reversed{T}$ contains a vertex in $R$. 
Similarly, if $u$ lies in $R$, then by \cref{thm:sampling}, with constant probability, $T$ contains a vertex in $L$. 
Thus, we only have to consider the pairs $(s, t) \in T \times \{u\}\, \cup\, \{u\} \times \reversed{T}$.
For each of these $O(\kappa)$ pairs, we use the incremental exact max-flow data structure of~\cite{DBLP:journals/jal/Henzinger97,GuptaK21} to report an $(s, t)$-vertex cut of size $\cutsize$ until no such cut exists.

\begin{restatable}[\textsc{IncrementalMaxFlowVC}~\cite{DBLP:journals/jal/Henzinger97,GuptaK21}]{theorem}{thmMaxFlowVC} \label{thm:maxflow_vc}
    Given a graph $G$ with $m$ edges and an integer $\cutsize$ such that $G$ is $\cutsize$-connected, as well as two vertices $s, t \in V(G)$, there is a data structure that maintains $G$ under $\ell$ edge insertions and reports an $(s, t)$-vertex cut of size $\cutsize$ after every insertion, or the fact that no such cut exists.
    The total running time is $O(\cutsize (m + \ell))$.
\end{restatable}

\begin{algorithm}[t]
\algrenewcommand{\algorithmicloop}{\textbf{repeat}}
    \caption{\textsc{IncrementalVC}} \label{alg:main}
    \begin{algorithmic}[1]
        \State Process edge insertions until $G$ is strongly connected. \label{line:reachability}
        \For{$\cutsize = 1, 2, 3, \dots, \floor{\frac n 2}$} \label{line:kappa_loop}
            \State $\boosts \gets \begin{cases} \ceil{\kappa \cdot (33 \log n + 288)} &\text{if the adversary is randomness-adaptive}\\ \ceil{27 \log n} &\text{if the adversary is oblivious}  \end{cases}$
            \For{each vertex $v \in V(G)$} \label{line:deg_check}
                \While{$\dgout{v} = \cutsize$ or $\dgin{v} = \cutsize$}
                    \State Report the $\cutsize$-vertex cut $\nbrout{v}$ or $\nbrin{v}$, respectively, and process edge insertions. \label{line:deg_report}
                \EndWhile
            \EndFor
            \For{$\volest = 1, 2, 4, 8, \dots, 2^{\floor{\log \frac{m}{48 \cutsize}}}$} \label{line:nu_loop}
                \Loop{} for both $G$ and the reversed graph $\reversed{G}$
                    \Loop{} $\boosts$ times \label{line:first_boost}
                        \State Run $\textsc{SampleRoots}(\cutsize, \volest)$ to obtain a set of terminals $T$. \label{line:loop_start}
                        \ForAll{$u \in T$} \label{line:terminal_loop}
                            \State Initialize the data structure for $\textsc{IncrementalRootedVC}(u, \cutsize, 2 \volest)$.
                            \While{$\textsc{IncrementalRootedVC}(u, \cutsize, 2 \volest)$ reports a $\cutsize$-vertex cut} \label{line:local_vc}
                                \State Process edge insertions with this data structure. \label{line:first_process} \label{line:loop_end}
                            \EndWhile
                        \EndFor
                    \EndLoop
                \EndLoop
            \EndFor
            \For{$i = 1, \dots, \boosts$} \label{line:middle_boost}
                \State Run $\textsc{SampleRoots}_{G}(\cutsize, \frac{m}{48 \cutsize})$ to obtain a set of terminals $T_i$.
                \State Run $\textsc{SampleRoots}_{\reversed{G}}(\cutsize, \frac{m}{48 \cutsize})$ to obtain a set of terminals $\reversed{T_i}$.
                \State Run $\textsc{UnbiasedSample}()$ to obtain a terminal $u_i$. \label{line:sampling_end}
            \EndFor
            \For{$i = 1, \dots, \boosts$} \label{line:boost_flows}
                \ForAll{$(s, t) \in T_i \times \{u_i\}$} \label{line:first_flow_loop}
                    \State Initialize the data structure for $\textsc{IncrementalMaxFlowVC}_{G}(s, t, \cutsize)$. \label{line:first_flow_start}
                    \While{$\textsc{IncrementalMaxFlowVC}_{G}(s, t, \cutsize)$ reports a $\cutsize$-vertex cut} \label{line:first_flow}
                        \State Process edge insertions with this data structure. \label{line:second_process} \label{line:first_flow_end}
                    \EndWhile
                \EndFor
                \ForAll{$(s, t) \in \reversed{T_i} \times \{u_i\}$} \label{line:second_flow_loop}
                    \State Initialize the data structure for $\textsc{IncrementalMaxFlowVC}_{\reversed{G}}(s, t, \cutsize)$. \label{line:second_flow_start}
                    \While{$\textsc{IncrementalMaxFlowVC}_{\reversed{G}}(s, t, \cutsize)$ reports a $\cutsize$-vertex cut} \label{line:second_flow}
                        \State Process edge insertions with this data structure. \label{line:third_process} \label{line:second_flow_end}
                    \EndWhile
                \EndFor
            \EndFor
        \EndFor
    \end{algorithmic}
\end{algorithm}

Finally, let us address the case where the initial graph $G$ is not strongly connected.
In this case, we pick an arbitrary vertex $u$ and use a simple folklore algorithm (see, e.g., \cite{italiano1986})
to maintain the sets of vertices reachable from $u$ in $G$ and $\reversed{G}$, respectively,
with total running time $O(m)$.
The graph $G$ is strongly connected as soon as every vertex is reachable from $u$ in both $G$ and $\reversed{G}$.
For completeness, we state this result in the following lemma, which we prove in \cref{sec:aux_datastructs}.
We outline the complete algorithm in \cref{alg:main}.
\begin{restatable}[Incremental 1-connectivity]{lemma}{lemmIncConn} \label{lemm:incReachability}
    Given a graph $G$ with $n$ vertices, there is a data structure that maintains $G$ under edge insertions and, after every insertion, reports whether the graph is strongly connected or not.
    The total running time is $O(m)$, where $m$ is the number of edges in $G$ after all insertions.
\end{restatable}

\paragraph{Subtle Pitfalls} Let us point out three important subtleties with \cref{alg:main}:
First, it is important that the for loop over $\nu$ goes from small to large.
A separation triple $(L, S, R)$ with small initial $\volout{L}$ might acquire a large $\volout{L}$ due to edge insertions.
(This is a nasty possibility that does not occur for static algorithms.)
Our property of being $(\kappa, \nu)$-ready is monotone under edge insertions, so we do not miss any such triple.
Similarly, for a separation triple $(L, S, R)$, the property whether
$\volout{L}$ or $\volin{R}$ is bigger can change back and forth due to edge insertions,
so it is important that we consider both $G$ and $\reversed{G}$ for a given value of $\nu$.

Second, it is important that the sampling on \cref{line:loop_start} takes place inside the for loop over $\nu$.
An algorithm that first performs all its sampling upfront and then spends $\Omega(\volout{L})$ time to detect a separation triple $(L, S, R)$
also struggles with the case where $\volout{L}$ starts out small and then drastically increases:
At first, $\volout{L}$ is small, so we would need a large number of samples to hit some vertex in $L$. 
Then, if $\volout{L}$ greatly increases, we can no longer afford to spend $\Omega(\volout{L})$ time on each of those samples.

Third, the value of $m$ also increases with edge insertions. The for loop over $\volest$ in \cref{alg:main}
runs until $\nu > \frac{m}{48 \kappa}$ where $m$ denotes the current number of edges.
When this happens, it is important to do all the sampling (\crefrange{line:middle_boost}{line:sampling_end}) at this very moment,
as future edge insertions could increase $m$ to the point where the inequality $\nu > \frac{m}{48 \kappa}$ no longer holds.

\paragraph{Oblivious vs.\ Adaptive Adversary} Our algorithm is Monte-Carlo randomized,
and works on both directed and undirected graphs if the adversary is oblivious:
for every value of $\kappa$, it suffices to consider the final remaining $\kappa$-separation triple $(L, S, R)$
and apply \cref{thm:sampling,thm:local_vc} to this particular cut.
If we only have to output the current value of $\kappa$,
then an output-adaptive adversary is equivalent to an oblivious one,
as the output yields no information about the internal state of our algorithm.
Indeed, the value of $\kappa$ is uniquely determined by the current graph.

However, this argument no longer works if the adversary is randomness-adaptive or we have to output the vertices of some $\kappa$-separation triple:
an adaptive adversary could first learn which vertices we sampled via \cref{thm:sampling}
or what the internal state of \cref{thm:local_vc} is,
    pick a $\kappa$-separation triple $(L, S, R)$ that our sampling misses (or that \cref{thm:local_vc} misses)
    and then make this the final remaining $\kappa$-separation triple.
For \emph{undirected graphs}, \cref{lemm:triple_count} shows that, if $G$ is \emph{$(\kappa, \nu)$-ready}, then there are $O(n^{\kappa+1})$ triples $(L, S, R)$
that are \volcut{\kappa}{<}{2\nu}s.
Therefore, we can repeat the sampling $O(\kappa \log n)$ times and union bound over all such triples.
This argument does not extend to \emph{directed graphs}, however: while there are $O(n^{\kappa})$ choices for $S$, there might be $\binom{n}{\nu-1} = n^{\Theta(\nu)}$ many choices for $L$.
This would require $\Omega(\nu)$ repetitions, which is too slow.
Thus, our algorithm only works on undirected graphs in this setting.

\subsection{Lower Bound Overview}

Our lower bounds are conditional on the Online Matrix-Vector Multiplication conjecture (OMv).

\begin{definition}[OMv problem~\cite{DBLP:conf/stoc/HenzingerKNS15}]
    In the OMv problem, you are first given an $n \times n$ boolean matrix $M$ and you may spend some time preprocessing $M$.
    Then, you are given $n$ boolean vectors $v_1, \dots, v_n$ of size $n$ one by one.
    For each vector $v_i$, you have to output the product $M v_i$ before seeing the next vector $v_{i+1}$.
    The computation time is the total time spent after seeing the first vector.
\end{definition}

\begin{conjecture}[\cite{DBLP:conf/stoc/HenzingerKNS15}]
    There is no randomized algorithm for OMv with success probability at least $\frac 1 2$, $O(\poly(n))$ preprocessing time and $O(n^{3 - \varepsilon})$ computation time for any $\varepsilon > 0$.
\end{conjecture}

There is an existing reduction from OMv to the incremental or decremental maximum cardinality bipartite matching problem.

\begin{theorem}[Theorem 1 in \cite{Dahlgaard16}] \label{theo:matching_lower_bound}
    There is no randomized algorithm for solving incremental (or decremental) maximum cardinality bipartite matching
    with success probability at least $\frac 1 2$, amortized $O(n^{1-\varepsilon})$ insertion (or deletion) time and
    $O(n^{2-\varepsilon})$ query time for any $\varepsilon > 0$, unless the OMv conjecture is false.
    
    More precisely, under OMv there is no such algorithm for processing $\Theta(n^2)$ edge insertions (or deletions) and $\Theta(n)$ queries on a graph with $\Theta(n)$ vertices in $O(n^{3-\varepsilon})$ total time,
    even if the size of the maximum cardinality matching is $\Theta(n)$ throughout.
\end{theorem}

We design a graph transformation that takes a bipartite graph $G$ and produces a (non-bipartite) graph $H$
such that a minimum vertex cover in $G$ corresponds to a minimum vertex cut in $H$.
By K\"onig's theorem, the size of a minimum vertex cover in a bipartite graph is equal to the size of a maximum bipartite matching.
This shows a reduction from incremental maximum matching to incremental minimum vertex cut.
We apply this reduction to \cref{theo:matching_lower_bound} to get the following lower bound:

\begin{restatable}{theorem}{theoLowerBound} \label{theo:lower_bound}
    There is no incremental (or decremental) randomized algorithm that maintains the vertex connectivity of an undirected graph
    with success probability at least $\frac 1 2$, amortized $O(n^{1-\varepsilon})$ time per insertion (or deletion) and
    amortized $O(n^{2-\varepsilon})$ time per query for any $\varepsilon > 0$, unless the OMv conjecture is false.
    
    More precisely, under OMv there is no such algorithm for processing $\Theta(n^2)$ edge insertions (or deletions) and $\Theta(n)$ queries on a graph with $\Theta(n)$ vertices in $O(n^{3-\varepsilon})$ time,
    even if the vertex connectivity is $\Theta(n)$ throughout.
\end{restatable}

To obtain a lower bound in the restricted case where the size of a minimum cut may not exceed $\kappa$,
we apply \cref{theo:matching_lower_bound} with $n \in \Theta(\kappa)$ and add additional vertices and edges to the graph, without introducing any new small vertex cuts, until we have $m$ edge insertions (or deletions).
\cref{theo:lower_bound} shows that the total running time of an incremental (or decremental) algorithm cannot be $O(\kappa^{3-\varepsilon})$.
This time can be distributed among all edge insertions, so we get the following bound:

\begin{restatable}{theorem}{theoLowerBoundKappa} \label{theo:lower_bound_kappa}
    There is no incremental (or decremental) algorithm that maintains the vertex connectivity of an undirected graph
    with amortized time $O(\kappa^{3-\varepsilon} / m)$ per insertion (or deletion) and
    $O(\kappa^{2-\varepsilon})$ time per query for any $\varepsilon > 0$, unless the OMv conjecture is false.
\end{restatable}

\subparagraph{Directed Graphs} The lower bounds in \cref{theo:lower_bound,theo:lower_bound_kappa} also apply to the directed case.
Indeed, the vertex connectivity of an undirected graph is equal to the vertex connectivity of the directed graph obtained by replacing each edge $uv$ by two directed edges $u \to v, v \to u$.
This increases the number of edge insertions or deletions by a constant factor.
Thus, an incremental (or decremental) algorithm for directed vertex connectivity implies such an algorithm for the undirected case, with the same asymptotic running time.

% BUILD_COMMAND make main.pdf -B
\section{Lower Bound Details} \label{sec:lower_bound}

% In this section, we show that the OMv conjecture implies that there is no efficient incremental (or decremental) algorithm for maintaing the vertex connectivity.
% \begin{theorem} \label{theo:lower_bound}
%     There is no incremental (or decremental) algorithm for maintaining the vertex connectivity of a graph
%     with amortized time $O(n^{1-\varepsilon})$ per insertion (or deletion) and
%     $O(n^{2-\varepsilon})$ time per query unless the OMv conjecture is false.
% \end{theorem}
% 
% We show a reduction from incremental maximum bipartite matching
% to incremental vertex connectivity.
% Then, we apply an existing reduction from OMv to incremental maximum bipartite matching.
% 
% \begin{theorem}[Theorem 1 in Dahlgaard \cite{Dahlgaard16}] \label{theo:matching_lower_bound}
%     There is no algorithm for solving incremental (or decremental) maximum cardinality bipartite matching
%     with amortized time $O(n^{1-\varepsilon})$ per insertion (or deletion) and
%     $O(n^{2-\varepsilon})$ time per query unless the OMv conjecture of \cite{Henzinger2015} is false.
% \end{theorem}

As mentioned in the overview, we want to reduce from incremental (or decremental) maximum bipartite matching
to incremental (or decremental) minimum vertex cut.
The core of this reduction is the following graph construction, which we also illustrate with an example in \cref{fig:lower_bound_transformation}:
For any bipartite graph $G$ with sides $X, Y$, construct a graph $H$ as follows:
Add two additional vertices $\hat{x}$ and $\hat{y}$, and then add edges to make $X \cup \{\hat{x}\}$ and $Y \cup \{\hat{y}\}$ into cliques.
Note that $G$ is naturally a subgraph of $H$, so a set of vertices in $G$ is also a set of vertices in $H$.

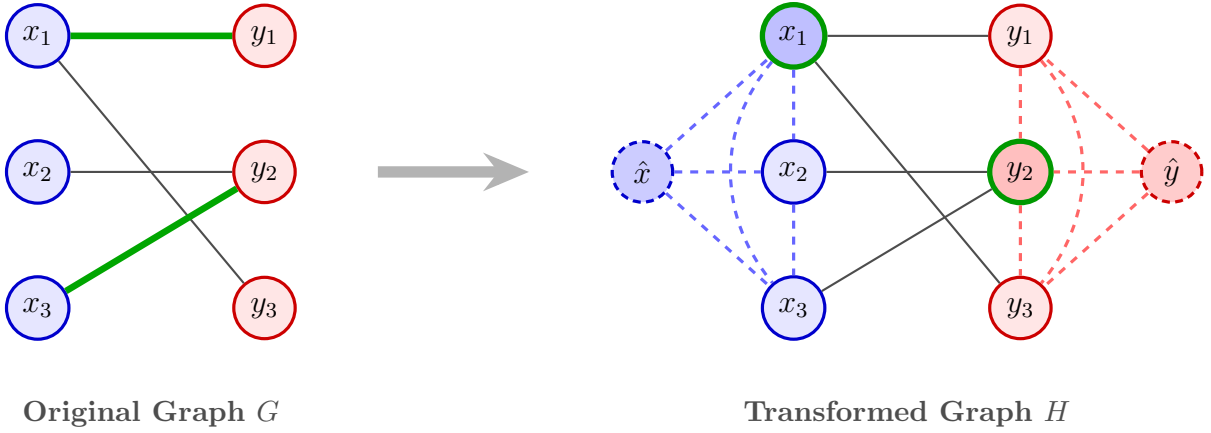
\begin{figure}
    \centering
    \begin{tikzpicture}[
    % Global settings
    >=Stealth,
    % Node Styles
    v1/.style={circle, draw=blue!80!black, fill=blue!10, very thick, minimum size=8mm, font=\large},
    v2/.style={circle, draw=red!80!black, fill=red!10, very thick, minimum size=8mm, font=\large},
    v1new/.style={v1, fill=blue!20, densely dashed},
    v2new/.style={v2, fill=red!20, densely dashed},
    % Koenig's Theorem Highlight Styles (Vertex Cover)
    v1cover/.style={v1, draw=green!60!black, fill=blue!25, line width=2pt},
    v2cover/.style={v2, draw=green!60!black, fill=red!25, line width=2pt},
    % Edge Styles
    bipedge/.style={thick, draw=black!70},
    clique1/.style={very thick, draw=blue!60, dashed},
    clique2/.style={very thick, draw=red!60, dashed},
    % Koenig's Theorem Highlight Style (Matching)
    matchedge/.style={line width=2.5pt, draw=green!65!black},
    % Arrow Style
    transition/.style={-{Stealth[length=6mm, width=5mm]}, line width=1.5mm, draw=black!30}
]

% ==========================================
%            ORIGINAL GRAPH G
% ==========================================

% V1 (Left side) 
\node[v1] (u1) at (0, 1.8) {$x_1$};
\node[v1] (u2) at (0, 0) {$x_2$};
\node[v1] (u3) at (0, -1.8) {$x_3$};

% V2 (Right side)
\node[v2] (v1) at (3.0, 1.8) {$y_1$};
\node[v2] (v2) at (3.0, 0) {$y_2$};
\node[v2] (v3) at (3.0, -1.8) {$y_3$};

% Bipartite edges designed for a max matching of size 2
\draw[matchedge] (u1) -- (v1); % Matching Edge 1
\draw[bipedge]   (u1) -- (v3);
\draw[bipedge]   (u2) -- (v2);
\draw[matchedge] (u3) -- (v2); % Matching Edge 2

% Subcaption 
\node[font=\bfseries, text=black!70] at (1.5, -3.2) {Original Graph $G$};

% ==========================================
%           TRANSITION ARROW
% ==========================================
\draw[transition] (4.5, 0) -- (6.5, 0);

% ==========================================
%          TRANSFORMED GRAPH H
% ==========================================

% V1 Transformed (Left side) 
\node[v1new]   (u4_t) at (8.0, 0) {$\hat x$}; 
\node[v1cover] (u1_t) at (10.0, 1.8) {$x_1$}; % Minimum Vertex Cover Node
\node[v1]      (u2_t) at (10.0, 0) {$x_2$};
\node[v1]      (u3_t) at (10.0, -1.8) {$x_3$};

% V2 Transformed (Right side) 
\node[v2]      (v1_t) at (13.0, 1.8) {$y_1$};
\node[v2cover] (v2_t) at (13.0, 0) {$y_2$}; % Minimum Vertex Cover Node
\node[v2]      (v3_t) at (13.0, -1.8) {$y_3$};
\node[v2new]   (v4_t) at (15.0, 0) {$\hat y$}; 

% Transformed Bipartite edges (mirroring the original, but standard style)
\draw[bipedge] (u1_t) -- (v1_t);
\draw[bipedge] (u1_t) -- (v3_t);
\draw[bipedge] (u2_t) -- (v2_t);
\draw[bipedge] (u3_t) -- (v2_t);

% New Clique 1 edges (V1 side - Blue)
\draw[clique1] (u1_t) -- (u2_t);
\draw[clique1] (u2_t) -- (u3_t);
\draw[clique1] (u1_t) to[bend right=40] (u3_t); 
\draw[clique1] (u4_t) -- (u1_t);
\draw[clique1] (u4_t) -- (u2_t);
\draw[clique1] (u4_t) -- (u3_t);

% New Clique 2 edges (V2 side - Red)
\draw[clique2] (v1_t) -- (v2_t);
\draw[clique2] (v2_t) -- (v3_t);
\draw[clique2] (v1_t) to[bend left=40] (v3_t); 
\draw[clique2] (v4_t) -- (v1_t);
\draw[clique2] (v4_t) -- (v2_t);
\draw[clique2] (v4_t) -- (v3_t);

% Subcaption
\node[font=\bfseries, text=black!70] at (11.5, -3.2) {Transformed Graph $H$};

\end{tikzpicture}
    \caption{The graph transformation used in our lower bound.
    The size of a maximum matching in the bipartite graph $G$ corresponds to the size of a minimum vertex cover in $G$, which in turn corresponds to the size of a minimum vertex cut in $H$.
    A maximum matching in $G$ and a minimum vertex cut in $H$ are marked green.}
    \label{fig:lower_bound_transformation}
\end{figure}

\begin{lemma} \label{lemm:cover_cut}
    A set $S \subseteq V(G)$ is a vertex cover in $G$ if and only if $S$ is a vertex cut in $H$.
\end{lemma}
\begin{proof}
    We first show that every vertex cover in $G$ is a vertex cut in $H$.
    Let $S \subseteq V(G)$ be a vertex cover. Then, every edge in $G$ is incident to at least one vertex in $S$.
    In particular, in $H \setminus S$, there are no edges between $(X \cup \{\hat{x}\}) \setminus S$ and $(Y \cup \{\hat{y}\}) \setminus S$.
    Moreover, each of these sets contains at least one vertex, namely $\hat{x}$ and $\hat{y}$, respectively.
    Thus, $S$ is a vertex cut in $H$.

    Conversely, we now show that every vertex cut in $H$ that avoids the vertices $\hat{x}, \hat{y}$
    is a vertex cover in $G$.
    Let $S \subseteq V(G)$ be a vertex cut in $H$.
    We show that every edge in $G$ is incident to a vertex in $S$.
    Indeed, suppose for contradiction that there is an edge $\{x, y\} \in E(G)$ with $x, y \notin S$.
    But then $H \setminus S$ is connected, as it consists of two cliques $(X \cup \{\hat{x}\}) \setminus S$ and $(Y \cup \{\hat{y}\}) \setminus S$
    connected by at least one edge $\{x, y\}$. But $S$ was assumed to be a vertex cut in $H$.
    This contradiction shows that there is no such edge $\{x, y\}$, so $S$ is a vertex cover.
\end{proof}

\begin{lemma} \label{lemm:matching_cut}
    The vertex connectivity of $H$ is equal to the size of a maximum matching in $G$.
\end{lemma}
\begin{proof}
    Since $G$ is bipartite, by K\"onig's theorem, the size of a maximum matching in $G$ is equal to the size of a minimum vertex cover in $G$.
    Let $S$ be a minimum vertex cover in $G$.
    By \cref{lemm:cover_cut}, the set $S$ is a vertex cut in $H$.
    Thus, the vertex connectivity of $H$ is at most the size of a minimum vertex cover in $G$.
    
    Conversely, let $S \subseteq V(H)$ be a minimum vertex cut.
    Then $H - S$ has two or more connected components.
    Since $X$ is a vertex cut in $H$, we have $|S| \le |X|$.
    Suppose for contradiction that $\hat{x} \in S$.
    Then, there is a vertex $x' \in X - S$.
    Since $x'$ has edges to all of $X - S - x'$,
    all vertices in $X - S$ lie in a single connected component of $H - S$.
    Since all neighbors of $\hat{x}$ lie in $X - S$,
    $H - (S - \hat{x})$ has the same number of connected components as $H - S$,
    so $S - \hat{x}$ is a vertex cut, of smaller cardinality than $S$.
    This contradiction shows that $\hat{x} \notin S$.
    Similarly, $\hat{y} \notin S$.
    Therefore $S \subseteq V(G)$, so by \cref{lemm:cover_cut}, $S$ is a vertex cover in $G$.
    Thus, the vertex connectivity of $H$ is at least the size of a minimum vertex cover in $G$.
\end{proof}

We now show the lower bound with no restrictions on $\kappa$.
\theoLowerBound*
\begin{proof}
    Suppose that there is an incremental (or decremental) algorithm $\cA$ for maintaining the vertex connectivity of a graph.
    We design an incremental (or decremental) algorithm $\cA'$ for maintaining the size of a maximum matching as follows:
    Given the initial graph $G$, construct $H$ as described above and initialize $\cA$ on $H$.
    For every edge insertion (or deletion), perform the corresponding insertion (or deletion) on $\cA$.
    For every maximum matching query, perform a vertex connectivity query in $\cA$ and output the result.
    By \cref{lemm:matching_cut}, this algorithm $\cA'$ is correct.
    Moreover, $\cA'$ has the same insertion (or deletion) and query time as $\cA$.
    Thus, \cref{theo:matching_lower_bound} implies \cref{theo:lower_bound}.
\end{proof}

Finally, we consider the case where the vertex connectivity may not exceed some small value $\kappa$.
For $m \in \Theta(\kappa^2)$, setting $n = \kappa$ in \cref{theo:lower_bound} yields a lower bound
of $\Omega(\kappa^{1-\varepsilon})$ per insertion (or deletion) or $\Omega(\kappa^{2-\varepsilon}) = \Omega(m^{1-\varepsilon})$ per query.
However, if we consider the total running time (instead of the amortized running time over a prefix of operations), this opens the door to a faster algorithm when $m \gg \kappa^2$.
(Note that, in a $\kappa$-connected graph, we always have $m \in \Omega(\kappa n)$ and $n > \kappa$, which forces $m \in \Omega(\kappa^2)$.)

We modify our graph construction to accommodate $m > \kappa^2$ edge insertions as follows:
Given a bipartite graph $G$ with $2 \kappa$ vertices, we construct the graph $H$ as described above.
Then, we replace the vertex $\hat{x}$ by $\Theta(\sqrt{m})$ copies with the same neighbors as $\hat{x}$,
so that these copies form an independent set (incremental setting) or a clique (decremental setting).
This yields a new graph $H'$ where we can perform $\Theta(m)$ edge insertions (or deletions) on these copies,
and crucially, the vertex connectivity of $H'$ will never exceed $\kappa$ since deleting $Y$ disconnects $\hat{x}$ from $\hat{y}$.
\cref{lemm:matching_cut} also holds for $H'$ since every vertex in $X$ is connected to all other vertices in $X$ and all copies of $\hat{x}$.

\theoLowerBoundKappa*

\begin{proof}
    Suppose that there is an incremental (or decremental) algorithm $\cA$ for maintaining the vertex connectivity of a graph.
    We design an incremental (or decremental) algorithm $\cA'$ for maintaining the size of a maximum matching as follows:
    Given the initial $2 \kappa$-vertex graph $G$, construct $H'$ as described above and initialize $\cA$ on $H'$.
    For every edge insertion (or deletion), perform the corresponding insertion (or deletion) on $\cA$.
    For every maximum matching query, perform a vertex connectivity query in $\cA$ and output the result.
    Then, perform $m$ (useless) edge insertions (or deletions) on the copies of $\hat{x}$.
    By \cref{lemm:matching_cut}, this algorithm $\cA'$ is correct.
    By \cref{theo:matching_lower_bound}, the total computation time of $\cA'$ cannot be $O(\kappa^{3 - \varepsilon})$.
    Since we perform $\Theta(m)$ insertions (or deletions) and $\kappa$ queries,
    this implies \cref{theo:lower_bound_kappa}.
\end{proof}

\section{Analysis of \cref{alg:main}} \label{sec:analysis}

In this section, we first analyze the running time of \cref{alg:main}
under the assumption that the algorithm is correct.
Then, we show that the algorithm produces a correct output with high probability.

\subsection{Running time analysis}

\begin{theorem}
    Let $\kappa$ be the vertex connectivity of the final graph $G$.
    If all outputs of \cref{alg:main} were correct,
    then the total running time of \cref{alg:main} over $m$ edge insertions on an initially empty graph is: $O(m \kappa^3 \log^2 n)$ if the adversary is oblivious; and $O(m \kappa^4 \log^2 n)$ if the adversary is randomness-adaptive.
\end{theorem}
\begin{proof}
    Processing edge insertions until $G$ is strongly connected takes $O(m)$ time in total by \cref{lemm:incReachability}.
    Since we assume that all outputs of \cref{alg:main} were correct,
    the for loop on \cref{line:kappa_loop} executes $O(\kappa)$ times. We call one iteration of this loop a \emph{round}.

    We start by bounding the cumulative running time of \crefrange{line:deg_check}{line:deg_report} over all rounds.
    By maintaining the degree of each vertex in an array,
    we can check or update the degree of each vertex in constant time, while reporting a trivial vertex cut takes $O(\cutsize)$ time.
    In each round, we perform a degree check at least once for every vertex, as well as after each edge insertion.
    Since we have at most $m$ edge insertions across $O(\kappa)$ rounds, these checks take $O(\kappa n + m) \subseteq O(m)$ time in total. 
    All reports occurs right before the next edge insertion, so they take $O(\kappa m)$ time in total.
    Hence, \crefrange{line:deg_check}{line:deg_report} cumulatively contribute $O(\kappa m)$ to the running time.

    We now analyze the cumulative running time of \crefrange{line:nu_loop}{line:loop_end} over all rounds. Consider one particular iteration of the repeat loop on \cref{line:first_boost}.
    By \cref{thm:sampling}, invoking $\textsc{SampleRoots}(\cutsize, \volest)$ takes $O(m)$ time and
    yields a set of terminals $T$ with $|T| \in O(\frac{m}{\nu})$.
    By \cref{thm:local_vc}, a single $\textsc{IncrementalRootedVC}(s, \cutsize, 2 \volest)$ structure takes $O(\nu \kappa^2 + \ell)$ time in total, where $\ell$ is the number of edge insertions that happen while this structure exists.
    Since only one of these structures is active at a time, the $O(\ell)$ terms sum to $O(m)$.
    Thus, the running time of one iteration of the repeat loop on \cref{line:first_boost} is $O(m + |T| \cdot \nu \kappa^2 + m) \subseteq O(m \kappa^2)$.
    In every round, it is repeated $\boosts$ times for both the reversed and original graph for $O(\log n)$ different values of $\nu$ (loop in \cref{line:nu_loop}), resulting in a total of $O(\kappa \cdot \boosts \cdot \log n)$ iterations.
    If the adversary is oblivious,
    then $\boosts \in O(\log n)$, so the number of iterations is $O(\kappa \log^2 n)$ and the total cumulative running time of \crefrange{line:nu_loop}{line:loop_end} is $O(m \kappa^3 \log^2 n)$.
    If the adversary is randomness-adaptive,
    then $\boosts \in O(\kappa \log n)$, so the number of iterations is $O(\kappa^2 \log^2 n)$ and the total cumulative running time of \crefrange{line:nu_loop}{line:loop_end} is $O(m \kappa^4 \log^2 n)$.

    Next, we analyze \crefrange{line:middle_boost}{line:sampling_end}.
    By \cref{thm:sampling,lemm:unbiased_sampling}, the calls to $\textsc{SampleRoots}(\cutsize, \frac{m}{48 \cutsize})$ and $\textsc{UnbiasedSample}()$ take $O(m)$ time.
    Across all rounds, these lines get executed $\kappa \cdot \boosts$ times,
    so their total contribution to the running time is $O(m \kappa^2 \log n)$.
    
    Finally, we analyze \crefrange{line:boost_flows}{line:second_flow_end}.
    By \cref{thm:sampling}, we have $|T_i| \in O(\kappa)$, $|\reversed{T_i}| \in O(\kappa)$.
    By \cref{thm:maxflow_vc}, a single $\textsc{IncrementalMaxFlowVC}(s, t, \cutsize)$ structure takes $O(\kappa m)$ time in total.
    Thus, a single iteration of either the for loop on \cref{line:first_flow_loop} or the for loop on \cref{line:second_flow_loop} takes $O(\kappa m)$ time.
    Across all rounds, the total number of such iterations is $O(\kappa \cdot \boosts \cdot (|T_i| + |\reversed{T_i}|)) = O(\kappa^2 \cdot \boosts)$.
    If the adversary is oblivious,
    then $\boosts \in O(\log n)$, so the number of iterations is $O(\kappa^2 \log n)$ and the total cumulative running time of \crefrange{line:boost_flows}{line:second_flow_end} is $O(m \kappa^3 \log n)$.
    If the adversary is randomness-adaptive,
    then $\boosts \in O(\kappa \log n)$, so the number of iterations is $O(\kappa^3 \log n)$ and the total cumulative running time of \crefrange{line:boost_flows}{line:second_flow_end} is $O(m \kappa^4 \log n)$.
\end{proof}

\subsection{Correctness}

We now show the correctness of \cref{alg:main}, assuming that \cref{thm:sampling,thm:local_vc,lemm:unbiased_sampling,thm:maxflow_vc} hold.
We assume without loss of generality that the graph is initially empty.
If the initial graph is not empty, then we insert its edges one by one, and ignore the cuts output by the algorithm.

We show results both for an oblivious adversary that has to fix a sequence of insertions in advance and a randomness-adaptive adversary that can change its mind depending on the internal randomness our algorithm uses.
At any point in time, the interaction between our algorithm and the adversary can be described by a sequence $e_1, \dots, e_\ell$
of edges that have been inserted, and a sequence of cuts $S_1, \dots, S_\ell$ that our algorithm has output.
We show that $S_i$ is a minimum vertex cut of the graph with edges $e_1, \dots, e_i$ with high probability.
We can model a randomness-adaptive adversary as a function $\cA$ that takes the sequence $e_1, \dots, e_\ell$, as well as the sequence of random bits used by the algorithm so far,
and outputs a new edge $e_{\ell+1}$ that is inserted.

\begin{definition}
    We say that \cref{alg:main} \emph{succeeds} if, after every edge insertion, it outputs a set $S$ that is a minimum vertex cut in the current graph.
\end{definition}
\begin{theorem}[Correctness of \cref{alg:main}] \label{thm:main_correct}
    Given a graph $G$ on $n$ vertices undergoing $m$ edge insertions supplied by an adversary $\cA$,
    \cref{alg:main} succeeds with probability at least $1 - n^{-5}$,
    when either (1) $G$ is directed or undirected, $\cA$ is oblivious, and $\boosts \ge 27 \log n$ or (2) $G$ is undirected, $\cA$ is randomness-adaptive, and $\boosts \ge \kappa \cdot (33 \log n + 288)$.
\end{theorem}

We split the proof of \cref{thm:main_correct} into two parts: Correctly maintaining the current vertex connectivity (\cref{thm:kappa_correct}),
and correctly outputting a vertex cut of that size (\cref{lemm:set_correct}).

\begin{lemma} \label{lemm:set_correct}
    Whenever \cref{alg:main} outputs some set $S$ of vertices, $S$ is a vertex cut and $|S| = \kappa$, where $\kappa$ is the variable of the for loop on \cref{line:kappa_loop}.
\end{lemma}
\begin{proof}
    \cref{alg:main} only outputs a set $S$ on \cref{line:deg_report,line:local_vc,line:first_flow,line:second_flow}.
    Each of these lines lies in a while loop that terminates as soon as the underlying routine (finding a trivial vertex cut, $\textsc{IncrementalRootedVC}$, or $\textsc{IncrementalMaxFlowVC}$) fails to report a $\kappa$-vertex cut.
\end{proof}

\begin{theorem} \label{thm:kappa_correct}
    Fix some iteration $k$ of the for loop on \cref{line:kappa_loop} of \cref{alg:main},
    and assume that the graph is $k$-vertex connected at the start of the iteration.
    Then, at the end of the iteration, with probability at least $1 - n^{-6}$, the graph is $(k + 1)$-vertex connected.
\end{theorem}

Before proving \cref{thm:kappa_correct}, we show how it implies \cref{thm:main_correct} by union bounding over all iterations of the outer for loop.
\begin{proof}[Proof of \cref{thm:main_correct} using \cref{lemm:set_correct,thm:kappa_correct}]
    Let $\cutsize$ denote the vertex connectivity of the final graph.
    Let $\event_{i}$ be the good event that the graph is $i$-vertex connected at the start of the $i$-th iteration
    of the for loop on \cref{line:kappa_loop},
    and let $\event = \event_1 \cap \dots \cap \event_{\kappa}$.
    \Cref{line:reachability} ensures that the graph is $1$-vertex connected at the start of the first iteration.
    Thus the good event $\event_1$ occurs with probability $1$.
    \cref{thm:kappa_correct} shows that $\Pr[\event_{i} \mid \event_{i - 1}] \geq 1 - n^{-6}$ for $2 \leq i \le \cutsize$. Therefore,
    \[
    \Pr[\ol{\event_i} \cap \event_{i-1}] =  \Pr[\ol{\event_i} \mid \event_{i-1}] \cdot \Pr[\event_{i-1}] \le \Pr[\ol{\event_i} \mid \event_{i-1}] = 1 - \Pr[\event_{i} \mid \event_{i - 1}] \le n^{-6}.
    \]
    If $\ol{\event}$ happens, at least one of the events $\ol{\event_1}$ or $\ol{\event_i} \cap \event_{i-1}$ for $2 \le i \le \cutsize$ happens. By the union bound,
    \[
    \Pr[\ol{\event}] \le \Pr[\ol{\event_1}] + \sum_{i=2}^{\cutsize} \Pr[\ol{\event_i} \cap \event_{i-1}] \le 0 +  (\cutsize-1) \cdot n^{-6} \le n^{-5},
    \]
    since $\kappa \le n$. Hence, $\Pr[\event] \ge 1 - n^{-5}$

    If the event $\event$ happens, then at any point during the execution of \cref{alg:main},
    the vertex connectivity of the current graph is equal to the value $i$ of the variable of the for loop on \cref{line:kappa_loop}.
    By \cref{lemm:set_correct}, \cref{alg:main} always outputs a vertex cut of size $i$.
    Thus, the algorithm succeeds whenever the event $\event$ happens.
\end{proof}

What remains is to show \cref{thm:kappa_correct}.
For the remainder of this section, we fix some iteration of the for loop on \cref{line:kappa_loop} of \cref{alg:main} with parameter $\cutsize$, and assume that the graph is $\cutsize$-vertex connected at the start of the iteration.
We split the proof into two parts, first reasoning about the first part of the round, that is, the for loop on \cref{line:nu_loop}, and then reasoning about the second part of the round, that is, the for loop on \cref{line:boost_flows}.
In particular, we will show the following two lemmas.

\begin{restatable}{lemma}{lemmNuLoop} \label{lemm:correct_nu_loop}
    Suppose that either (1) $G$ is directed or undirected, $\cA$ is oblivious, and $\boosts \ge 27 \log n$ or (2) $G$ is undirected, $\cA$ is randomness-adaptive, and $\boosts \ge 33 \kappa \log n$.
    Then, with probability at least $1 - n^{-7}$, the for loop on \cref{line:nu_loop} of \cref{alg:main} only terminates once both the graph and the reversed graph are \ready{\cutsize}{\frac{m}{48 \cutsize}}.
\end{restatable}
\begin{restatable}{lemma}{lemmFlowLoop} \label{lemm:correct_flow_loop}
    Suppose that either (1) $G$ is directed or undirected, $\cA$ is oblivious, and $\boosts \ge 27 \log n$ or (2) $G$ is undirected, $\cA$ is randomness-adaptive, and $\boosts \ge \kappa \cdot (33 \log n + 288)$.
    Assume that the graph and the reversed graph are \ready{\cutsize}{\frac{m}{48 \cutsize}} at the start of the for loop on \cref{line:boost_flows} of \cref{alg:main}.
    Then, with probability at least $1 - n^{-7}$, this for loop only terminates once the graph is $(\cutsize + 1)$-vertex connected.
\end{restatable}

\begin{proof}[Proof of \cref{thm:kappa_correct} from \cref{lemm:correct_nu_loop,lemm:correct_flow_loop}]
    By \cref{lemm:correct_nu_loop}, we have that with probability at least $1 - n^{-7}$, the for loop on \cref{line:nu_loop} of \cref{alg:main} only terminates once the graph and the reversed graph are \ready{\cutsize}{\frac{m}{48 \cutsize}}.
    Condition on this event.
    Then, by \cref{lemm:correct_flow_loop}, we have that with probability at least $1 - n^{-7}$, the for loop on \cref{line:boost_flows} of \cref{alg:main} only terminates once the graph is $(\cutsize + 1)$-vertex connected.
    By union bounding over these two events, we have that with probability at least $1 - 2 n^{-7} \geq 1 - n^{-6}$, the good event that the graph is $(\cutsize + 1)$-vertex connected by the end of the iteration occurs, as desired.
\end{proof}

\subsection{Proof of \cref{lemm:correct_nu_loop}}
We first analyze a single iteration of the repeat loop on \cref{line:first_boost}.

\begin{claim} \label{clm:correct_single_cut}
    Let $G'$ be a directed graph on $n$ vertices and $m$ edges and let $\cA$ be a randomness-adaptive adversary.
    Fix an integer $\volest \leq \frac{m}{48 \cutsize}$,
    and let $(L, S, R)$ be any \volcut{\cutsize}{<}{2 \volest} in $G'$ with $\volout[G']{L} \geq \volest$.
    Then, with probability at least $\frac{1}{4}$, after executing \crefrange{line:loop_start}{line:loop_end} of \cref{alg:main} on $G'$ with $\cA$, the triple $(L, S, R)$ is no longer a valid \volcut{\cutsize}{<}{2 \volest}, that is, either $\volout[G']{L} \geq 2 \volest$, or $(L, S, R)$ is no longer a separation triple because an edge from $L$ to $R$ has been inserted.
\end{claim}
\begin{proof}
    We first note that, after \crefrange{line:deg_check}{line:deg_report}, we have $\mindegout[G] > \kappa$ and $\mindegout[\reversed{G}] = \mindegin[G] > \kappa$. Therefore, we may use \cref{thm:local_vc} and the \textsc{IncrementalRootedVC} data structure here.
    
    Let $\event$ denote the event that the call to \textsc{SampleRoots} on \cref{line:loop_start} of \cref{alg:main} returns a set $T$ with $L \cap T \neq \emptyset$.
    By \cref{thm:sampling}, $\Pr[\event] \ge \frac{1}{2}$.
    Crucially, this holds for any randomness-adaptive adversary $\cA$, since this sampling step happens before any edge insertions are performed on $G'$, and since it only depends on the current graph $G'$ (and not on what future edge insertions are performed).
    
    Suppose that $\event$ happens and let $s \in L \cap T$ be a vertex in this intersection.
    Then, by \cref{thm:local_vc}, with probability at least $\frac{1}{2}$, for any randomness-adaptive adversary $\cA$, the call to \textsc{IncrementalRootedVC} on \cref{line:local_vc} with argument $s$ will only output $\bot$ once $(L, S, R)$ is no longer a valid \volvertcut{\cutsize}{<}{2 \volest}{s}. 
    Hence, in the iteration of the for loop on \cref{line:terminal_loop} with $u = s$, the while loop on \cref{line:local_vc} does not terminate until $(L, S, R)$ is no longer a valid \volcut{\cutsize}{<}{2 \volest}.
    By the chain rule of probability, we thus have that with probability at least $\frac{1}{4}$, \crefrange{line:loop_start}{line:loop_end} of \cref{alg:main} will only terminate once $(L, S, R)$ is no longer a valid \volcut{\cutsize}{<}{2 \volest}, as desired.
\end{proof}

With this claim, we now proceed by showing that the for loop on \cref{line:first_boost} boosts the probability of success high enough so that we can union bound later.
\begin{claim} \label{clm:correct_first_boost}
    Let $G'$ be a directed graph on $n$ vertices and $m$ edges and let $\cA$ be a randomness-adaptive adversary.
    Fix some arbitrary $\volest \leq \frac{m}{48 \cutsize}$, and let $(L, S, R)$ be any \volcut{\cutsize}{\geq}{\volest} in $G'$.
    Then, with probability at least $1 - 2^{-\boosts/3}$, after executing the for loop on \cref{line:first_boost} of \cref{alg:main} on $G'$ with $\cA$, the triple $(L, S, R)$ is no longer a valid \volcut{\cutsize}{<}{2 \volest}, that is, either $\volout[G']{L} \geq 2 \volest$ or an edge from $L$ to $R$ has been inserted.
\end{claim}
\begin{proof}
    We execute the for loop on \cref{line:first_boost} of \cref{alg:main} on $G'$ with $\cA$.
    Let $\event_i$ be the bad event that after the $i$-th iteration of this for loop,
    $(L, S, R)$ is still a valid \volcut{\cutsize}{<}{2 \volest}.
    By \cref{clm:correct_single_cut}, we have that $\Pr[\event_1] \leq \frac{3}{4}$ and $\Pr[\event_i \mid \event_{i - 1}] \leq \frac{3}{4}$ for each $i \geq 2$.
    By the chain rule of probability, we thus have that
    \[
        \Pr[\event_{\boosts}] = \Pr[\event_1] \cdot \prod_{i = 2}^{\boosts} \Pr[\event_i \mid \event_{i - 1}] \leq \left(\frac{3}{4}\right)^{\boosts} \leq 2^{-\boosts / 3},
    \]
    which concludes the proof.
\end{proof}

Next, we show that each iteration of the for loop on \cref{line:nu_loop} preserves the invariant that both $G$ and $\reversed{G}$ are \ready{\cutsize}{\volest} with high probability.
We first consider the case of an oblivious adversary.
Crucially, we use the fact that edge insertions cannot create new \volcut{\cutsize}{<}{2\volest}s.

\begin{claim} \label{clm:correct_nu_iteration_oblivious}
    Consider some iteration $i$ of the for loop on \cref{line:nu_loop} of \cref{alg:main} with parameters $\cutsize$ and $\volest \leq \frac{m}{48 \cutsize}$ on a \emph{directed or undirected graph $G$ with an oblivious adversary $\cA$}, with $\boosts \ge 27 \log n$.
    Assume further that the graph $G$ and the reversed graph $\reversed{G}$ are \ready{\cutsize}{\volest} at the start of the iteration.
    Then, with probability at least $1 - n^{-8}$, the good event $\event_{\cutsize, \volest}$ that the graph and the reversed graph are \ready{\cutsize}{2 \volest} by the end of the iteration occurs.
\end{claim}

\begin{proof}
    Let $G'$ be either $G$ or $\reversed{G}$.
    We use $G'_0$ to denote the state of $G'$ at the start of the iteration $i$,
    and $G'_1$ to denote the state of $G'$ after the iteration.
    By assumption, $G'_0$ is \ready{\cutsize}{\volest}, and, in particular, $\kappa$-connected.
    Since $G'_0$ is a spanning subgraph of $G'_1$, the graph $G'_1$ is also $\kappa$-connected.
    In a $\kappa$-connected graph, the property of being a \volcut{\cutsize}{<}{2 \volest} is monotone under edge insertions:
    a triple $(L, S, R)$ can stop being a \volcut{\cutsize}{<}{2 \volest} because it is no longer a $\kappa$-separation triple or because $\volout{L}$ becomes too big, but a triple cannot suddenly start being one.

    Since the adversary $\cA$ is oblivious, it will perform the same sequence of edge insertions on $G'_0$ independent of what the algorithm does.
    Among all \volcut{\cutsize}{<}{2 \volest}s in $G'_0$, let $(L, S, R)$ be a triple that remains a \volcut{\cutsize}{<}{2 \volest}
    for the longest under this sequence of edge insertions.
    Then, $G'_1$ has a \volcut{\cutsize}{<}{2 \volest} if and only if $(L, S, R)$ is a \volcut{\cutsize}{<}{2 \volest} in $G'_1$.

    By \cref{clm:correct_first_boost}, we have that, with probability at least $1 - 2^{-\boosts/3}$, the for loop of \cref{line:first_boost} of \cref{alg:main} run on $G'$ will only terminate once $(L, S, R)$ is not a valid \volcut{\cutsize}{<}{2 \volest}.
    Since $\boosts \ge 27 \log n$, we have $1 - 2^{-\boosts/3} \ge 1 - n^{-9}$.

    By union bounding over the two graphs $G$ and $\reversed{G}$, we thus have that with probability at least $1 - 2 n^{-9} \geq 1 - n^{-8}$, both the graph and the reversed graph are \ready{\cutsize}{2 \volest} by the end of the iteration, as desired.
\end{proof}

If $\cA$ is a randomness-adaptive adversary, then we can no longer speak of the \volcut{\cutsize}{<}{2 \nu} that survives the longest.
Instead, we need to bound the number of \volcut{\cutsize}{<}{2 \nu}s,
and then union bound over all these triples. Here, our approach only works for undirected graphs, since for directed graphs there may be exponentially many choices for the sets $L$ and $R$.

\lemmTripleCount*
\begin{proof}
    If $(L, S, R)$ is a \volcut{\cutsize}{<}{2 \nu} in an undirected graph $G$,
    then: $|S| = \kappa$, the sets $L, R$ are non-empty, $G \setminus S$ is disconnected, $\vol{L} < 2 \nu$, and there are no edges between $L$ and $R$.
    Since $|S| = \kappa$, there are at most $\binom{n}{\kappa} \le n^{\kappa}$ choices for the set $S$.
    What remains is to show that, for each set $S$, there are at most $n$ choices for the set $L$.
    
    Since there are no edges between $L$ and $R$, the set $L$ is the union of one or more connected components of $G \setminus S$.
    For each such component $C_i$, the triple $(C_i, S, V(G) \setminus (C_i \cup S))$ is a $\kappa$-separation triple in $G$.
    Therefore, since $G$ is $(\kappa, \nu)$-ready, each such component $C_i$ has $\vol{C_i} \ge \nu$.
    Since $\vol{L} < 2 \nu$, the set $L$ consists of exactly one connected component.
    The number of components of $G \setminus S$ is at most $n$, so there are also at most $n$ choices for the set $L$.
\end{proof}

\cref{lemm:triple_count} fails on directed graphs: while there are $n^\cutsize$ choices for $S$, there might be exponentially many choices for $L$.
For example, consider a 4-cycle $u, v, w, x, u$ and let $\cutsize < \nu < n$.
Replace $u$ and $w$ by $\nu$ copies that form a (directed) clique; $v$ by $n$ copies that form an independent set, and $x$ by $\kappa$ copies that form a clique. See \cref{fig:many_directed}.
The resulting graph $G$ has $\Theta(n)$ vertices and a unique minimum vertex cut $S$ that consists of all copies of $x$.
For the corresponding separating triples $(L, S, R)$, all copies of $w$ lie in $L$, and all copies of $u$ lie in $R$, but each copy of $v$ may lie in either $L$ or $R$.
Thus, the number of $\kappa$-separating triples is $2^n$, and the number of  \volcut{\cutsize}{<}{2 \nu}s is $\binom{n}{\nu-1}$.

\begin{figure}
    \centering
    \includegraphics[width=0.5\textwidth]{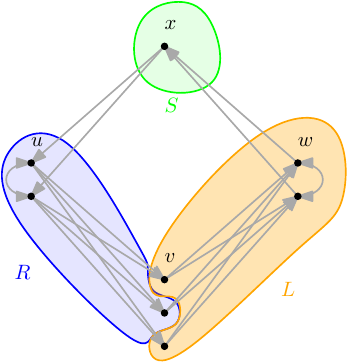}
    \caption{\cref{lemm:triple_count} does not extend to directed graphs. This graph has many separating triples for the same $S$, as each copy of $v$ may lie in $L$ or $R$.}
    \label{fig:many_directed}
\end{figure}

\begin{claim} \label{clm:correct_nu_iteration_adaptive}
    Consider some iteration $i$ of the for loop on \cref{line:nu_loop} of \cref{alg:main} with parameters $\cutsize$ and $\volest \leq \frac{m}{48 \cutsize}$ on an \emph{undirected graph $G$ with a randomness-adaptive adversary $\cA$}, with $\boosts \ge 33 \kappa \log n$.
    Assume further that the graph $G$ and the reversed graph $\reversed{G}$ are \ready{\cutsize}{\volest} at the start of the iteration.
    Then, with probability at least $1 - n^{-8}$, the good event $\event_{\cutsize, \volest}$ that the graph and the reversed graph are \ready{\cutsize}{2 \volest} by the end of the iteration occurs.
\end{claim}

\begin{proof}
    Let $G'$ be either $G$ or $\reversed{G}$. We use $G'_0$ to denote the state of $G'$ at the start of the iteration $i$,
    and $G'_1$ to denote the state of $G'$ after the iteration.
    By assumption, $G'_0$ is \ready{\cutsize}{\volest}, and, in particular, $\kappa$-connected.
    Since $G'_0$ is a spanning subgraph of $G'_1$, the graph $G'_1$ is also $\kappa$-connected.
    As in the proof of \cref{clm:correct_nu_iteration_oblivious}, we observe that the property
    of being a \volcut{\cutsize}{<}{2 \volest} is monotone under edge insertions.
    Thus, a triple $(L, S, R)$ can only be a \volcut{\cutsize}{<}{2 \volest} in $G'_1$ if it is a \volcut{\cutsize}{<}{2 \volest} in $G'_0$.

    By \cref{clm:correct_first_boost}, we have that with probability at least $1 - 2^{-\boosts/3}$, the for loop of \cref{line:first_boost} of \cref{alg:main} run on $G'$ will only terminate once $(L, S, R)$ is not a valid \volcut{\cutsize}{<}{2 \volest} for any fixed \volcut{\cutsize}{\geq}{\volest} $(L, S, R)$ in $G'$.
    Since $\boosts \ge 33 \cutsize \log n$, we have $1 - 2^{-\boosts / 3} \ge 1 - n^{-(\cutsize+10)}$.
    By \cref{lemm:triple_count}, at most $n^{\cutsize+1}$ triples $(L, S, R)$ are \volcut{\cutsize}{<}{2 \volest}s in $G'_0$.
    Thus, we can union bound over all such triples to obtain that with probability at least $1 - n^{-9}$, the for loop of \cref{line:first_boost} of \cref{alg:main} run on $G'$ will only terminate once there are no more valid \volcut{\cutsize}{<}{2 \volest}s, that is, once $G'$ is \ready{\cutsize}{2 \volest}.

    By union bounding over the two graphs $G$ and $\reversed{G}$, we thus have that with probability at least $1 - 2 n^{-9} \geq 1 - n^{-8}$, both the graph and the reversed graph are \ready{\cutsize}{2 \volest} by the end of the iteration, as desired.
\end{proof}

Finally, we show \cref{lemm:correct_nu_loop} by union bounding over all iterations of the for loop on \cref{line:nu_loop}.

\lemmNuLoop*

\begin{proof}
    Before the for loop on \cref{line:nu_loop} of \cref{alg:main}, we trivially have that the graph and the reversed graph are \ready{\cutsize}{1}, and so we have $\Pr[\event_{\cutsize, 1}] \geq 1 - n^{-8}$ by \cref{clm:correct_nu_iteration_oblivious} (or \cref{clm:correct_nu_iteration_adaptive}).
    That claim also implies that $\Pr[\event_{\cutsize, \volest} \mid \event_{\cutsize, \volest / 2}] \geq 1 - n^{-8}$ for each $\volest \in \{2, 4, 8, \ldots, 2^{\floor{\log \frac{m}{48 \cutsize}}}\}$.
    Notice again that each iteration uses independent randomness, and so our sequence of events forms a Markov chain, meaning that $\Pr[\event_{\cutsize, 2^i} \mid \event_{\cutsize, 2^{i - 1}} \cap \ldots \cap \event_{\cutsize, 1}] = \Pr[\event_{\cutsize, 2^i} \mid \event_{\cutsize, 2^{i - 1}}]$ for all $i \in \{1, \ldots, \log \frac{m}{48 \cutsize}\}$.
    We can use the law of total probability with this Markov property to obtain that
    \[
        \Pr[\event_{\cutsize, 2^{\floor{\log \frac{m}{48 \cutsize}}}}] \geq \Pr[\event_{\cutsize, 1}] \cdot \prod_{i = 1}^{\floor{\log \frac{m}{48 \cutsize}}} \Pr[\event_{\cutsize, 2^i} \mid \event_{\cutsize, 2^{i-1}}] \geq (1 - n^{-8})^{\log \frac{m}{48 \cutsize}} \geq 1 - n^{-7},
    \]
    where the last inequality follows from Bernoulli's inequality and the fact that $\log m \leq n$.
    As $2^{\floor{\log \frac{m}{48 \cutsize}}} \geq \frac{m}{96 \cutsize}$, with probability at least $1 - n^{-7}$, the for loop on \cref{line:nu_loop} of \cref{alg:main} only terminates once the graph and the reversed graph are \ready{\cutsize}{2 \cdot \frac{m}{96 \cutsize}}, i.e., \ready{\cutsize}{\frac{m}{48 \cutsize}}.
\end{proof}

\subsection{Proof of \cref{lemm:correct_flow_loop}}

The last missing piece now is to prove \cref{lemm:correct_flow_loop}.
That is, we now assume that both the graph and the reversed graph are \ready{\cutsize}{\frac{m}{48 \cutsize}} and discuss the for loop on \cref{line:boost_flows} of \cref{alg:main}.
We will first argue that, with high probability, the sampling in the for loop on \cref{line:middle_boost} hits at least one vertex from each of $L$ and $R$ for every $\cutsize$-separation triple $(L, S, R)$.
We will then appeal to \cref{thm:maxflow_vc} to show that the for loop on \cref{line:boost_flows} will only terminate once the graph is $(\cutsize + 1)$-vertex connected with high probability.

\begin{claim} \label{clm:correct_flow_sampling}
    Assume that both the graph $G$ and its reversed graph are \ready{\cutsize}{\frac{m}{48 \cutsize}}.
    Fix a $\cutsize$-separation triple $(L, S, R)$ in $G$.
    With probability at least $1-2^{-\boosts/3}$, there exists an iteration $i$ of the for loop on \cref{line:middle_boost} such that either $(T_i \times \{u_i\}) \cap (L \times R) \neq \emptyset$ or $(\reversed{T_i} \times \{u_i\}) \cap (R \times L) \neq \emptyset$.
\end{claim}
\begin{proof}
    First, fix some $\cutsize$-separation triple $(L, S, R)$ in $G$, and for each $1 \le i \leq \boosts$ define the good event $\event_i$ that $(T_i \times \{u_i\}) \cap (L \times R) \neq \emptyset$ or $(\reversed{T_i} \times \{u_i\}) \cap (R \times L) \neq \emptyset$.
    As both $G$ and its reversed graph are \ready{\cutsize}{\frac{m}{48 \cutsize}}, we have that $\volout{L} \geq \frac{m}{48 \cutsize}$ and $\volout[\reversed{G}]{R} \geq \frac{m}{48 \cutsize}$.
    We can thus use \cref{thm:sampling} to obtain that with probability at least $\frac{1}{2}$, we have $T_i \cap L \neq \emptyset$, and also with probability at least $\frac 1 2$, we have $\reversed{T_i} \cap R \neq \emptyset$.
    Similarly, we can use \cref{lemm:unbiased_sampling} to obtain that with probability at least $\frac{1}{2}$, we will have $u_i \in (L \cup R)$.
    
    Now since the sets $T_i$, $\reversed{T_i}$ and $\{u_i\}$ are sampled independently for each iteration $i$, we have that
    \[
        \Pr[\event_i \mid u_i \in R] \geq \Pr[T_i \cap L \neq \emptyset] \geq \frac 1 2,
    \]
    and similarly
    \[
        \Pr[\event_i \mid u_i \in L] \geq \Pr[\reversed{T_i} \cap R \neq \emptyset] \geq \frac 1 2.
    \]
    We can thus use the law of total probability to conclude that
    \[
        \Pr[\event_i] \geq \frac 1 2 \cdot (\Pr[u_i \in R] + \Pr[u_i \in L]) = \frac 1 2 \cdot \Pr[u_i \in (L \cup R)] \geq \frac 1 4.
    \]
    Now we can use independence of the iterations $i$ to conclude that
    \[
        \Pr\!\left[\bigcup_{i=1}^{\boosts} \event_i\right] = 1 - \prod_{i=1}^{\boosts} (1 - \Pr[\event_i]) \geq 1 - \left(\frac 3 4\right)^{\boosts} \geq 1 - 2^{-\boosts/3},
    \]
    that is, with probability at least $1 - 2^{-\boosts/3}$, there exists an iteration $i$ of the for loop on \cref{line:middle_boost} such that either $(T_i \times \{u_i\}) \cap (L \times R) \neq \emptyset$ or $(\reversed{T_i} \times \{u_i\}) \cap (R \times L) \neq \emptyset$.
\end{proof}

\begin{lemma} \label{lemm:triple_count_volume}
    If $G$ is an \emph{undirected} graph with $m$ edges that is $(\kappa, \nu)$-ready, then there are at most $n^{\kappa} \cdot 2^{2 m / \nu}$ triples $(L, S, R)$ that are $\kappa$-separation triples.
\end{lemma}
\begin{proof}
    The proof is similar to that of \cref{lemm:triple_count}.
    If $(L, S, R)$ is a $\kappa$-separation triple in an undirected graph $G$,
    then: $|S| = \kappa$, the sets $L, R$ are non-empty, $G \setminus S$ is disconnected, and there are no edges between $L$ and $R$.
    Since $|S| = \kappa$, there are at most $\binom{n}{\kappa} \le n^{\kappa}$ choices for the set $S$.
    What remains is to show that, for each set $S$, there are at most $2^{2 m / \nu}$ choices for the set $L$.
    
    Since there are no edges between $L$ and $R$, the set $L$ is the union of one or more connected components of $G \setminus S$.
    For each such component $C_i$, the triple $(C_i, S, V(G) \setminus (C_i \cup S))$ is a $\kappa$-separation triple in $G$.
    Since $G$ is $(\kappa, \nu)$-ready, each such component $C_i$ has $\volout{C_i} \ge \nu$.
    As $G$ is undirected, $\volout{G} = 2 m$, so in particular, there are at most $2 m / \nu$ such components.
    The set $L$ is uniquely determined by a subset of these components, so there are at most $2^{2 m / \nu}$ choices for the set $L$.
\end{proof}

\begin{claim} \label{clm:correct_flow_sampling_forall}
    Suppose that $G$ is undirected, $\cA$ is randomness-adaptive, and $\boosts \ge \kappa (33 \log n + 288)$.
    Assume that both the graph $G$ and its reversed graph are \ready{\cutsize}{\frac{m}{48 \cutsize}}.
    Then, with probability at least $1 - n^{-7}$, for every $\cutsize$-separation triple $(L, S, R)$ in $G$, there exists an iteration $i$ of the for loop on \cref{line:middle_boost} such that either $(T_i \times \{u_i\}) \cap (L \times R) \neq \emptyset$ or $(\reversed{T_i} \times \{u_i\}) \cap (R \times L) \neq \emptyset$.
\end{claim}
\begin{proof}
    We apply \cref{clm:correct_flow_sampling} to every $\cutsize$-separation triple $(L, S, R)$ in $G$ and union bound over all such triples.
    By \cref{lemm:triple_count_volume}, there are at most $n^{\kappa} \cdot 2^{2 \cdot 48 \kappa}$ such triples.
    For each triple, by \cref{clm:correct_flow_sampling}, with probability at least $1 - 2^{-\boosts/3} \ge 1 - 2^{-96 \kappa} \cdot n^{-11 \kappa}$,
    there exists an iteration $i$ of the for loop on \cref{line:middle_boost} such that either $(T_i \times \{u_i\}) \cap (L \times R) \neq \emptyset$ or $(\reversed{T_i} \times \{u_i\}) \cap (R \times L) \neq \emptyset$.
    By a union bound, the probability that this happens simultaneously for every triple is at least $1 - n^{\kappa} \, 2^{96 \kappa} \cdot 2^{-96 \kappa}\, n^{-11 \kappa} = 1 - n^{-10 \kappa} \ge 1 - n^{-7}$.
\end{proof}

Now we can finally show \cref{lemm:correct_flow_loop} and thus conclude the proof of correctness.
\lemmFlowLoop*
\begin{proof}
    Suppose that $(L, S, R)$ is a $\cutsize$-separation triple in $G$ for which there is an iteration $i$ of the for loop on \cref{line:middle_boost} such that either $(T_i \times \{u_i\}) \cap (L \times R) \neq \emptyset$ or $(\reversed{T_i} \times \{u_i\}) \cap (R \times L) \neq \emptyset$.
    We now consider this $i$-th iteration of the for loop on \cref{line:boost_flows}.

    If $(T_i \times \{u_i\}) \cap (L \times R) \neq \emptyset$, then there exist vertices $s \in T_i \cap L$ and $t = u_i$.
    Consider now the corresponding iteration of the for loop on \cref{line:first_flow_loop} with $s$ and $t$.
    By \cref{thm:maxflow_vc}, the call to \textsc{IncrementalMaxFlowVC} on \cref{line:first_flow} will report an $(s, t)$-vertex cut of size $\cutsize$ in $G$ while such a cut exists.
    Since $(L, S, R)$ is an $(s, t)$-separation triple of size $\cutsize$, this means that the for loop on \cref{line:first_flow_loop} will only terminate once $(L, S, R)$ is no longer a valid $\cutsize$-separation triple.

    By an equivalent argument, if $(\reversed{T_i} \times \{u_i\}) \cap (R \times L) \neq \emptyset$, then the for loop on \cref{line:second_flow_loop} will only terminate once $(L, S, R)$ is no longer a valid $\cutsize$-separation triple.
    In this case, the call to \textsc{IncrementalMaxFlowVC} on \cref{line:second_flow} will report an $(s, t)$-vertex cut of size $\cutsize$ in $\reversed{G}$ while such a cut exists for $s \in \reversed{T_i} \cap R$ and $t = u_i$.

    If $G$ is undirected, $\cA$ is randomness-adaptive, and $\boosts \ge \kappa (33 \log n + 288)$,
    then \cref{clm:correct_flow_sampling_forall} shows that, with probability at least $1 - n^{-7}$,
    this argument applies simultaneously to all $\cutsize$-separation triples $(L, S, R)$ in $G$.
    Hence, with probability at least $1 - n^{-7}$, the for loop on \cref{line:boost_flows} of \cref{alg:main} will only terminate once the graph is $(\cutsize + 1)$-vertex connected, as desired.

    If instead $\cA$ is oblivious and $\boosts \ge 27 \log n$,
    then let $(L, S, R)$ be the $\kappa$-separation triple in $G$ that survives for the longest.
    By \cref{clm:correct_flow_sampling}, with probability at least $1 - 2^{-\boosts/3} \ge 1 - n^{-7}$,
    this argument applies to this particular triple $(L, S, R)$.
    Hence, with probability at least $1 - n^{-7}$, the for loop on \cref{line:boost_flows} of \cref{alg:main} will only terminate once $(L, S, R)$ is no longer a $\kappa$-separation triple.
    But then, the graph is $(\cutsize + 1)$-vertex connected, as desired.
\end{proof}

\section{Sampling} \label{sec:sampling}

In this section, we prove \cref{thm:sampling,lemm:unbiased_sampling}, which analyze the subroutines \textsc{SampleRoots} and \textsc{UnbiasedSample}.
The first subroutine samples edges uniformly at random and returns their tail vertices, ensuring that any set of vertices with large volume is hit with constant probability.
\thmSampling*
\begin{proof}
    The algorithm samples a multiset $E'$ of $\floor{\frac{2m}{\volest}}$ edges independently and uniformly at random with replacement, and returns the set
    \[
        T = \{u \in V(G) \mid \exists (u,v) \in E'\}
    \]
    of vertices with some outgoing edge in $E'$.
    Now fix any set $A \subseteq V(G)$ with $\volout{A} \geq \volest$.
    Each sampled edge originates from $A$ independently with probability $\frac{\volout{A}}{m} \geq \frac{\volest}{m}$, so the probability that no sampled edge originates from $A$ is
    \[
        \left(1 - \frac{\volout{A}}{m}\right)^{\floor{\frac{2 m}{\volest}}} \leq \left(1 - \frac{\volest}{m}\right)^{\frac{m}{\volest}} \leq e^{-1} \leq \frac{1}{2}.
    \]
    As this is exactly the probability that $A \cap T = \emptyset$, we have established the claim.
\end{proof}
\iffalse
\begin{proof}[Proof of \cref{thm:sampling}]
    The algorithm runs the algorithm from \cref{clm:single_sampling} twice, obtaining the sets $T_1$ and $T_2$, and then returns the union $T = T_1 \cup T_2$.
    Now consider any separation triple $(L, S, R)$ with $\volest \leq \volout{L} \leq \volout{R}$ and $|S| = \cutsize$.
    By \cref{clm:single_sampling}, we have that $\Pr[L \cap T_1 = \emptyset] \leq \frac{1}{2}$ and $\Pr[L \cap T_2 = \emptyset] \leq \frac{1}{2}$, as well as $\Pr[R \cap T_1 = \emptyset] \leq \frac{1}{2}$ and $\Pr[R \cap T_2 = \emptyset] \leq \frac{1}{2}$.
    As the $T_i$ are chosen independently, we thus have
    \[
        \Pr[L \cap T = \emptyset] = \Pr[L \cap T_1 = \emptyset] \cdot \Pr[L \cap T_2 = \emptyset] \leq \frac{1}{4},
    \]
    and similarly $\Pr[R \cap T = \emptyset] \leq \frac{1}{4}$.
    By union bounding over these two events, we thus have that $\Pr[L \cap T = \emptyset \lor R \cap T = \emptyset] \leq \frac{1}{2}$, as desired.
    The size of $T$ is at most $\frac{4 m}{\volest}$, and the algorithm runs in $\cO(m)$ time.
\end{proof}
\fi

If we only want to hit some vertex on either side of the cut and the cut size is not too large, then we can simply sample a vertex uniformly at random.
\lemmUnbiased*
\begin{proof}
    Let $(L, S, R)$ be any separation triple with $|S| = \kappa \le \frac{n}{2}$. Then $|L \cup R| = n - |S| \ge \frac{n}{2}$,
    thus,
    \[
    \Pr[u \in (L \cup R)] = \frac{|L \cup R|}{n} \ge \frac{1}{2}.
    \]
\end{proof}

\section{Incremental Rooted Cuts} \label{sec:dyn_rooted_vc}

Mirroring the definition of a \volcut{\cutsize}{\le}{\volest},
a \emph{$(\le\!\cutsize, \le\! \volest)$-edge cut} is an edge cut $(L, R)$ with $|\bndout{L}| \le \cutsize$ and $\volout{L} \le \volest$.
\cite{DBLP:conf/soda/ForsterNYSY20} design a local algorithm for computing a small rooted vertex or edge cut.
Given a directed graph $G$, two integers $\cutsize, \volest$, and a vertex $s \in V(G)$,
their algorithm either computes an edge cut $(L, R)$ with $s \in L$, $|\bndout{L}| \le \cutsize$, and $\volout{L} \in \cO(\cutsize \volest)$,
or certifies that no \voledgecut{\le\!\cutsize}{\le}{\volest}{s} exists. Crucially, if $G$ is available in memory, their algorithm runs in $\cO(\volest \cutsize^2)$ time, with no dependence on $m$.
More precisely, they show:

\begin{theorem}[Corollary A.2 in \cite{DBLP:conf/soda/ForsterNYSY20}] \label{theo:static_local_edge_cuts}
    Given a directed graph $G$ with $m$ edges and two integers $\cutsize, \volest$ with $\volest \leq \frac{m}{8 \cutsize}$, as well as a vertex $s \in V(G)$, there is a randomized (Monte Carlo) algorithm that runs in $\cO(\volest\cutsize^2)$ time and outputs
    \begin{itemize}
        \item the symbol ``$\bot$'' indicating that, with probability $\frac{1}{2}$, there is no \voledgecut{\le\!\cutsize}{\le}{\volest}{s}, or
        \item an \voledgecut{\cutsize}{\le}{10\volest\cutsize}{s}.
    \end{itemize}
\end{theorem}

The idea behind this algorithm is to mimic the Ford--Fulkerson max-flow algorithm: we repeatedly grow a search tree from $s$ to find augmenting paths, but cap each search at $8 \cutsize \volest$ explored edges.
Whenever this cap is reached, we sample one of the explored edges uniformly at random and send one unit of flow along the tree path to the tail of this edge, i.e. reverse all edges on this path.
If instead a search gets stuck before reaching the cap, the algorithm returns the set of vertices reachable from $s$ in the residual graph as a certificate for a cut; and if no search gets stuck within $\cutsize + 1$ attempts, the algorithm reports that no cut has been found.

To see why this succeeds, suppose $G$ has an \voledgecut{\cutsize}{\le}{\volest}{s} $(L, R)$.
In a single iteration, the tail of the sampled edge lies in $R$ with probability $1 - \cO(\frac{1}{\cutsize})$, in which case the augmenting path crosses the cut $(L, R)$ and saturates one of its edges.
By a union bound, with constant probability, the first $\cutsize$ augmentations saturate every edge of $\bndout{L}$.
The next search then cannot leave $L$, so it gets stuck after exploring at most $\volout{L} \le \volest < 8 \cutsize \volest$ edges, and the algorithm returns the set of vertices reachable from $s$ in the residual graph as a certificate for an \voledgecut{\cutsize}{\le}{10\volest\cutsize}{s}.

There is a standard reduction from rooted (or $s$-$t$) vertex cut to rooted (or $s$-$t$) edge cut in directed graphs that involves splitting every vertex into two new vertices. By applying this reduction to \cref{theo:static_local_edge_cuts},
\cite{DBLP:conf/soda/ForsterNYSY20} obtain a local algorithm for computing rooted vertex cuts.

\subparagraph{Our contribution.} We extend their algorithm to the incremental setting. For directed edge cuts, we obtain a data structure that maintains such a cut under edge insertions by growing the search tree via a single-source incremental reachability structure. Our approach generalizes the incremental $s$-$t$ cut algorithm of \cite{GuptaK21} to the rooted, low-volume setting.

\begin{restatable}[\textsc{IncrementalRootedEC}]{theorem}{theoIncrementalEC} \label{thm:local_ec}
    Given a directed graph $G$ with $m$ edges and two positive integers $\cutsize, \volest$ with $\volest \leq \frac{m}{2 \cutsize}$, as well as a vertex $s \in V(G)$, there is a data structure that maintains $G$ under $\ell$ edge insertions supplied by a randomness-adaptive adversary $\cA$.
    After every insertion, the data structure reports either an \voledgecut{\cutsize}{\le}{3 \volest \cutsize}{s} or the symbol $\bot$.
    Let $(L, V(G) \setminus L)$ be one particular \voledgecut{\cutsize}{\le}{\volest}{s} in $G$.
    With probability $\ge \frac{1}{2}$, the data structure only outputs $\bot$ once $(L, V(G) \setminus L)$ is no longer a valid \voledgecut{\cutsize}{\le}{\volest}{s}.
    The total running time is $\cO(\volest  \cutsize^2 + \ell)$.
\end{restatable}

Then, via the aforementioned standard reduction, we obtain the subroutine for vertex cuts used by \cref{alg:main}:

\theoLocalVC*

\subsection{Edge Cuts}

\begin{algorithm}[t]
\algrenewcommand{\algorithmicloop}{\textbf{repeat}}
\caption{Rooted edge cut algorithm (Algorithm 2 in \cite{DBLP:conf/soda/ForsterNYSY20}).}
\label{alg:local_ec}
\begin{algorithmic}[1]
\newcommand{\Edfs}{E_{\mathrm{DFS}}}
\Function{LocalEC}{$s, \volest, \cutsize$}
    \Loop{} $\cutsize+1$ times
        \State Grow a DFS tree $T$ starting from $s$ and stop once exactly $8 \cutsize \volest$ edges have been visited.
        \State Let $E_{\mathrm{DFS}}$ be the set of edges visited.
        \If{$|\Edfs| < 8 \cutsize \volest$} \Return $V(T)$
        \Else
            \State Sample an edge $(u, v) \in \Edfs$ uniformly at random.
            \State Reverse the direction of edges in the path $P_{s,u}$ in $T$ from $s$ to $u$.
        \EndIf
    \EndLoop
    \Return $\bot$
\EndFunction
\end{algorithmic}
\end{algorithm}

The local algorithm of \cite{DBLP:conf/soda/ForsterNYSY20} for computing rooted edge cuts (\cref{theo:static_local_edge_cuts}) is depicted in \cref{alg:local_ec}.
The tree $T$ grown by this algorithm need not be a DFS tree: any way of growing a rooted directed tree (BFS, etc.) will do.
To make \cref{alg:local_ec} incremental, we grow a reachability tree $T$ under edge insertions as in \cref{clm:inc_reachability}.
Whenever an edge is inserted, we check whether it can be used to extend $T$, and if so, continue growing $T$ from there.
This way, we can ``pause'' \cref{alg:local_ec} whenever it outputs a cut,
and ``resume'' the algorithm whenever an edge is inserted.

The resulting algorithm is depicted in \cref{alg:inc_local_ec}.
We use $T(u)$ to denote the parent of a node $u$ in the reachability tree,
and explicitly maintain the set $L$ of vertices reachable from $s$ via this tree.
We use $H$ to denote the current state of the graph (after some edge reversals),
and $G$ to denote the original state of the graph.
Note that $H$ need not be simple: reversing an edge $(u,v)$ produces $(v,u)$, which may already be present in $H$, so $H$ can contain parallel edges and is thus a multigraph.
Since \cref{alg:inc_local_ec} restores $H$ to $G$ before returning $\bot$,
we need not spend $\Theta(m)$ time creating a fresh copy of $G$ for each of the many instances we run (in the context of \cref{alg:main});
instead, a single copy of $G$ can be passed to them one after another.

\begin{algorithm}[t]
\caption{Algorithm for incremental $s$-rooted $(\cutsize, \le\!\volest)$-edge cut}
\label{alg:inc_local_ec}
\begin{algorithmic}[1]
\Function{IncrementalRootedEC}{$G, s \in V(G), \cutsize, \volest$}
    \State $H \gets$ implicit copy of $G$
    \State Implicitly set $T(v) \gets \bot$ for all $v \in V(G)$
    \For{$k = 0, \dots, \cutsize$}
        \State $T(s)\gets\!\top$, $L \gets \{s\}$, $E_L \gets \emptyset$
        \State \Call{GrowReachabilityTree}{$\incout[H]{s}$}
        \While{$|E_L| < 2 \volest \cutsize$} \label{line:while}
            \State \textbf{Output} $L$
            \State $e = (u, v) \gets$ next edge insertion
            \State Add $e$ to $G$ and $H$
            \If{$T(u) \ne \bot$}
                \State \Call{GrowReachabilityTree}{$e$}
            \EndIf
        \EndWhile
        \State Sample an edge $(u,v) \in E_L$ uniformly at random
        \State Reverse the direction of edges on the path $P_{s,u}$ in $T$ from $s$ to $u$
        \State $T(v)\gets\!\bot$ for all $v \in L$
    \EndFor
    \State Undo all edge reversals done on $H$
    \State \textbf{Output} $\bot$
\EndFunction

\Statex

\Function{GrowReachabilityTree}{$A$}
    \For{$e=(u,v) \in A$ while $|E_L| < 2\volest\cutsize$}
        \State Add $e$ to $E_L$
        \If{$T(v) = \bot$}
        \State Add $v$ to $L$
            \State $T(v) \gets u$
            \State \Call{GrowReachabilityTree}{$\incout[H]{v}$}
        \EndIf
    \EndFor
\EndFunction
\end{algorithmic}
\end{algorithm}

To argue correctness of \cref{alg:inc_local_ec}, we need the following lemma from \cite{DBLP:conf/soda/ForsterNYSY20} about the effect of reversing a path on the cut size and volume.

\begin{lemma}[{\cite[Lemma 3.5]{DBLP:conf/soda/ForsterNYSY20}}]
\label{lem:path_reversal}
    Let $H$ be a directed multigraph, $L \subseteq V(H)$ a set with $s \in L$, and $P_{s,v}$ a path in $H$ from $s$ to some vertex $v$.
    Suppose we reverse the direction of every edge of $P_{s,v}$.
    If $v \notin L$, then $|\bndout[H]{L}|$ and $\volout[H]{L}$ each decrease by exactly one; otherwise, both remain unchanged.
\end{lemma}
\begin{proof}
    Reversing an edge that leaves $L$ turns it into an edge entering $L$, and vice versa; reversing an edge with both endpoints inside (or both outside) $L$ leaves the cut unaffected. Hence the change in $|\bndout[H]{L}|$ equals the number of times $P_{s,v}$ crosses from $V(H) \setminus L$ back to $L$ minus the number of times it crosses from $L$ to $V(H) \setminus L$.

    Since $P_{s,v}$ starts at $s \in L$, if $v \notin L$ it crosses from $L$ to $V(H) \setminus L$ exactly one more time than it crosses back, so reversing its edges decreases $|\bndout[H]{L}|$ by exactly one. If instead $v \in L$, the two kinds of crossings occur equally often, so $|\bndout[H]{L}|$ remains unchanged. The same argument applies to $\volout[H]{L}$, since a reversal never changes the number of edges with both endpoints in $L$.
\end{proof}

The following two lemmas establish the correctness of \cref{alg:inc_local_ec} in the incremental setting. They are analogous to \cite[Lemmas A.4 and A.5]{DBLP:conf/soda/ForsterNYSY20}, but since we do not require a $(1+\varepsilon)$-approximation of the cut size, our proofs are simpler and give a better bound on the volume.

\begin{lemma}
\label{lem:getcut_correct}
    Suppose that \cref{alg:inc_local_ec} (after any number of edge insertions) outputs $L$. Then $L \neq V(G)$, $|\bndout[G]{L}| \le \cutsize$, and $\volout[G]{L} \le 3\volest\cutsize$.
\end{lemma}
\begin{proof}
    We first show that $L \neq V(G)$. Since \cref{alg:inc_local_ec} only outputs a set $L$ inside the while loop on \cref{line:while}, we have $|E_L| < 2\volest\cutsize$.
    Therefore, in the current iteration of the outer for loop,
    the condition $|E_L| < 2\volest\cutsize$ was true in all calls to \textsc{GrowReachabilityTree},
    so $L$ is indeed the set of vertices reachable from $s$ in $H$, and $E_L$ contains all outgoing edges in $H$ from $L$.
    In particular, $|\bndout[H]{L}| = 0$ and $E_L = E_H(L, L)$.
    This is true even in the presence of edge insertions:
    if the inserted edge $(u,v)$ has $u \notin L$, then it cannot be reached from $s$, and if $u \in L$, then we call \textsc{GrowReachabilityTree}$(\{(u,v)\})$, which visits any newly reachable vertices and edges. Note that the check $T(u) \ne \bot$ in \cref{alg:inc_local_ec} is equivalent to $u \in L$.

    Since $H$ is obtained from $G$ by reversing edges, which preserves their endpoints, $|E_{G}(L, L)| = |E_H(L, L)| = |E_L| < 2\volest\cutsize \le m$, using $\volest \le \frac{m}{2\cutsize}$.
    In particular, we cannot have $L = V(G)$, as this would imply $|E_G(L,L)| = |E(G)| = m$.
    This shows $L \neq V(G)$.

    Consider $D \coloneqq |\bndout[G]{L}| - |\bndout[H]{L}|$ over the execution.
    Initially $H = G$, so $D$ starts out at $0$; an edge insertion adds its edge to $G$ and $H$ in the same orientation and leaves $D$ unchanged;
    and by \cref{lem:path_reversal}, reversing an augmenting path $P_{s,u}$ keeps $|\bndout[G]{L}|$ fixed while decreasing $|\bndout[H]{L}|$ by at most one, so it increases $D$ by at most one.
    As there are at most $\cutsize$ augmentations, $D = |\bndout[G]{L}| - |\bndout[H]{L}| \le \cutsize$. Since $|\bndout[H]{L}| = 0$ at the end of the search, we obtain $|\bndout[G]{L}| \le \cutsize$.

    Finally, $\volout[G]{L}$ counts the original edges with tail in $L$, namely those within $L$ together with those leaving $L$, so
    \[
        \volout[G]{L} = |E_G(L, L)| + |\bndout[G]{L}| = |E_L| + |\bndout[G]{L}| \le 2\volest\cutsize + \cutsize \le 3\volest\cutsize,
    \]
    using $\volest \ge 1$.
\end{proof}

\begin{lemma} \label{lemm:one_cut_detected}
    Let $(L, V(G) \setminus L)$ be one particular \voledgecut{\cutsize}{\le}{\volest}{s} in $G$.
    Then, with probability $\ge \frac{1}{2}$, \cref{alg:inc_local_ec} run on the graph $G$ with edge insertions supplied by a randomness-adaptive adversary $\cA$
    only outputs $\bot$ once $(L, V(G) \setminus L)$ is no longer a valid \voledgecut{\cutsize}{\le}{\volest}{s}.
\end{lemma}
\begin{proof}
    For the outer for loop of \cref{alg:inc_local_ec}, we say that the $i$'th iteration \emph{succeeds}
    if the tail of the sampled edge lies outside of $L$.
    We denote the event of this happening by $\event_i$,
    and the good event that the first $\cutsize$ iterations all succeed by $\event$.

    First we show that, if $\event$ happens, then the $(\cutsize+1)$'st iteration of the outer for loop does not terminate while
    $(L, V(G) \setminus L)$ is still a valid \voledgecut{\cutsize}{\le}{\volest}{s} in $G$;
    so in particular, \cref{alg:inc_local_ec} cannot output $\bot$ during this time.
    Suppose that $(L, V(G) \setminus L)$ is currently a valid \voledgecut{\cutsize}{\le}{\volest}{s} in $G$.
    By \cref{lem:path_reversal}, whenever an iteration of the outer for loop of \cref{alg:inc_local_ec} succeeds,
    the quantity $|\bndout[G]{L}| - |\bndout[H]{L}|$ increases by one.
    Furthermore, edge insertions do not change this quantity.
    Since $(L, V(G) \setminus L)$ is valid, we have $|\bndout[G]{L}| \le \cutsize$, and if $\event$ happens, then the first $\cutsize$ iterations all succeeded.
    Hence we have $|\bndout[G]{L}| - |\bndout[H]{L}| = \cutsize \ge |\bndout[G]{L}|$, forcing $|\bndout[H]{L}| = 0$.
    Therefore, in the $\cutsize+1$'st iteration, there are no edges in $H$ from $L$ to $V(G) \setminus L$.
    Thus, the tree $T$ is always contained in $L$, and at most $\volout[H]{L}$ edges can get added to $E_L$.
    Also, since $(L, V(G) \setminus L)$ is valid, $\volout[H]{L} \le \volout[G]{L} \le \volest$.
    Hence, in the $\cutsize+1$'st iteration of the outer for loop, the while loop on \cref{line:while} cannot terminate while $(L, V(G) \setminus L)$ is a valid \voledgecut{\cutsize}{\le}{\volest}{s} in $G$.
    
    What remains is to analyze the probability that $\event$ happens.
    To bound $\Pr[\event_i]$, let $E_L$ be the set of edges we sample from in iteration $i$, and let $A = \{(u,v) \in E_L \mid u \in L\}$ be the subset of those edges whose tail lies in $L$. Observe that $\Pr[\event_i] = 1 - \frac{|A|}{|E_L|}$.
    Every edge in $A$ is an out-edge of $L$ in $H$, so $|A| \le \volout[H]{L} \le \volout{L}$, where the second inequality holds because, by \cref{lem:path_reversal}, each augmentation only decreases $\volout[H]{L}$.
    Moreover, $|E_L| \ge 2\volest\cutsize$ due to the while loop on \cref{line:while}.
    Therefore, we have
    \[
        1 - \Pr[\event_i] = \frac{|A|}{|E_L|} \le \frac{\volout{L}}{|E_L|} \le \frac{\volest}{2\volest\cutsize} = \frac{1}{2\cutsize}.
    \]
    Note that the event $\event_i$ only depends on the current state of the graph $H$ at the time of the sampling,
    and not on what future edge insertions happen.
    Thus, this holds even if the adversary $\cA$ is randomness-adaptive.
    Finally, by the union bound, we have
    \[
    \Pr[\event] = \Pr[\event_1 \cap \dots \cap \event_\cutsize] \ge 1 - \sum_{i=1}^{\cutsize}(1 - \Pr[\event_i]) \ge 1 - \frac{\cutsize}{2 \cutsize} = \frac{1}{2}.
    \]

\end{proof}

Next, we analyze the running time of \cref{alg:inc_local_ec} in the incremental setting.

\begin{lemma}
\label{lem:local_ec_running_time}
    Running \cref{alg:inc_local_ec} for $\ell$ edge insertions takes $\cO(\volest \cutsize^2 + \ell)$ total time.
\end{lemma}
\begin{proof}
    We assume that $G$ is provided as doubly linked adjacency lists, and that $H$ is not an actual copy of $G$ but $G$ itself, with edges reversed in place.
    This way, an edge reversal can be performed in $\cO(1)$ time, and by keeping track of the pointer changes during an edge reversal, we can undo each edge reversal in $\cO(1)$ time as well.
    
    Similarly, we assume that the array $T$ is pre-initialized, that is $T(v) \gets \bot$ for all $v \in V(G)$, and we will restore $T$ to this state at the end of each iteration of the outer for loop.
    Also, whenever \textsc{GrowReachabilityTree} sets $T(v) \gets u$, we additionally store a pointer to the corresponding edge $(u, v)$ in $T(v)$'s adjacency list.
    This allows us to find the edge $(T(v), v)$ to be reversed in $\cO(1)$ time.

    Furthermore, we do not output the vertices of $L$ one by one; instead, we output a pointer to the linked list storing $L$, in $\cO(1)$ time.

    Consider one iteration of the outer for loop.
    Across this iteration, \textsc{GrowReachabilityTree} runs in $\cO(|E_L|)$ total time: each iteration of its for loop takes $\cO(1)$ time, adding one edge to $E_L$, and possibly adding one vertex to $L$ as well as updating its entry in $T$.
    Resetting $T$, $L$, and $E_L$ at the end of the iteration therefore also takes $\cO(|E_L|)$ time, as does sampling an edge $(u, v) \in E_L$ and reversing the tree path $P_{s,u}$.
    Since $|E_L| \le 2\volest\cutsize$, one iteration takes $\cO(\volest \cutsize)$ time, so the $\cutsize + 1$ iterations take $\cO(\volest \cutsize^2)$ time in total.
    These iterations reverse $\cO(\volest \cutsize^2)$ edges overall, so undoing all reversals at the end takes $\cO(\volest \cutsize^2)$ time as well.
    Finally, per edge insertion the algorithm spends only $\cO(1)$ time, to check whether $T(u) \ne \bot$ and to add the new edge to $G$ and $H$, which sums up to $\cO(\ell)$ time over all $\ell$ insertions.
    Altogether, the running time is $\cO(\volest \cutsize^2 + \ell)$.
\end{proof}

We are now ready to prove \cref{thm:local_ec}:
\theoIncrementalEC*
\begin{proof}
    By \cref{lem:getcut_correct}, if \cref{alg:inc_local_ec} outputs a set $L$ then $(L, V(G) \setminus L)$ is an \voledgecut{\cutsize}{\le}{3 \volest \cutsize}{s}.
    By \cref{lemm:one_cut_detected}, with probability at least $\frac{1}{2}$, the algorithm only outputs $\bot$ once
    one particular \voledgecut{\cutsize}{\le}{\volest}{s} is no longer valid, even if the adversary is randomness-adaptive.
    By \cref{lem:local_ec_running_time}, the total running time of \cref{alg:inc_local_ec} is $\cO(\volest \cutsize^2 + \ell)$.
\end{proof}

\subsection{Vertex Cuts}

We now use \cref{thm:local_ec} to prove \cref{thm:local_vc}.
Let $G$ be a directed graph with $n$ vertices and $m$ edges, and let $s \in V(G)$. We define the \emph{split graph} $G'$ of $G$ with respect to $s$ as follows. Every vertex $v \in V(G) \setminus \{s\}$ is split into an \emph{in-vertex} $v^-$ and an \emph{out-vertex} $v^+$ joined by an \emph{internal edge} $(v^-\!, v^+)$; every edge entering $v$ in $G$ now enters $v^-$, and every edge leaving $v$ now leaves $v^+$. Formally, $V(G') = \{v^-\!, v^+ \mid v \in V(G) \setminus \{s\}\} \cup \{s\}$ and $E(G') = \{(u^+\!, v^-) \mid (u,v) \in E(G)\} \cup \{(v^-\!, v^+) \mid v \in V(G) \setminus \{s\}\}$, where for convenience we write $s^- = s^+ = s$. Observe that $G'$ has $n' = 2n-1$ vertices and $m' = m+n-1$ edges.

The following two lemmas relate $G$ and $G'$, reducing the problem of finding a vertex cut in $G$ to that of finding an edge cut in $G'$. Their proofs are similar to those of \cite[Lemmas 4.4 and 4.5]{DBLP:conf/soda/ForsterNYSY20}, but improve the constants in the volume bounds.

\begin{lemma}
    \label{lem:vc_to_ec}
    For a directed graph $G$, a vertex $s \in V(G)$, and two positive integers $\cutsize, \volest$ with $\mindegout[G] > \cutsize$, let $(L, S, R)$ be an \volvertcut{\cutsize}{\le}{\volest}{s} in $G$, and let $G'$ be the split graph of $G$ with respect to $s$. Put $L' = \{v^-\!,v^+ \mid v \in L\} \cup \{v^- \mid v \in S\}$. Then $(L', V(G') \setminus L')$ is an \voledgecut{\cutsize}{\le}{2\volest}{s} in $G'$.
\end{lemma}
\begin{proof}
    We first show that $(L', V(G') \setminus L')$ is an $s$-rooted $\cutsize$-edge cut in $G'$.
    Since $(L, S, R)$ is a separation triple, $L$ has no edges to $R$, so every out-edge of $L$ ends in $L \cup S$; as the in-copies of $L$ and $S$ both lie in $L'$, no edge of $G'$ corresponding to an out-edge of $L$ in $G$ crosses the edge cut.
    The only edges in $\bndout[G']{L'}$ are therefore the internal edges $(v^-\!, v^+)$ with $v \in S$, of which there are exactly $|S| = \cutsize$.
    Hence $(L', V(G') \setminus L')$ is a $\cutsize$-edge cut, and $s \in L'$ since $s \in L$.

    Next we bound $\volout[G']{L'}$. Splitting $L'$ into its three parts,
    \[
        \volout[G']{L'} =  \underbrace{\sum_{v \in L} \dgout[G']{v^+}}_{A} + \underbrace{\sum_{v \in L \setminus \{s\}} \dgout[G']{v^-}}_{B} + \underbrace{\sum_{v \in S} \dgout[G']{v^-}}_{C}.
    \]
    The sum $A$ counts the out-edges of the vertices $\{v^+ \mid v \in L\}$, which are in bijection with the out-edges of the vertices of $L$ in $G$; hence $A = \volout[G]{L}$.
    The sums $B$ and $C$ count internal edges. Each in-vertex $v^-$ has out-degree $1$, so $B = |L| - 1$ and $C = |S|$.
    Using $\volout[G]{L} \le \volest$ and $|S| = \cutsize$,
    \[
        \volout[G']{L'} = A + B + C = \volout[G]{L} + (|L| - 1) + |S| \le \volest + |L| + \cutsize - 1.
    \]
    Since $|L| \ge 1$ and $\cutsize \ge 1$, we have $(|L| - 1)(\cutsize - 1) \ge 0$, that is, $|L| + \cutsize - 1 \le |L|\cutsize$.
    Moreover, since every vertex of $L$ has out-degree greater than $\cutsize$, we have $\volout[G]{L} = \sum_{v \in L} \dgout[G]{v} \ge |L|\cutsize$, so that $|L|\cutsize \le \volout[G]{L} \le \volest$.
    Combining these two bounds yields
    \[
        \volout[G']{L'} \le \volest + |L| + \cutsize - 1 \le \volest + |L|\cutsize \le 2\volest,
    \]
    which shows that $(L', V(G') \setminus L')$ is an \voledgecut{\cutsize}{\le}{2\volest}{s} in $G'$.
\end{proof}

\begin{lemma}
    \label{lem:ec_to_vc}
    For a directed graph $G$, a vertex $s \in V(G)$, and two positive integers $\cutsize, \volest$, let $G'$ be the split graph of $G$ with respect to $s$.
    Suppose that $\cutsize \le \frac{n}{2}$, $\mindegout[G] > \cutsize$, and $\volest \le \frac{m'}{4}$, and let $(L', V(G') \setminus L')$ be a \voledgecut{\cutsize}{\le}{\volest}{s} in $G'$.
    Let $L = \{v \in V(G) \mid \{v^-\!,v^+\} \subseteq L'\}$.
    Then $(L, \nbrout[G]{L}, V(G) \setminus (L \cup \nbrout[G]{L}))$ is an \volvertcut{\cutsize}{\le}{\volest}{s} in $G$.
\end{lemma}
\begin{proof}
    Note that $s \in L$, since $s \in L'$ and $s^- = s^+ = s$.
    We first bound the cut size. Fix $v \in \nbrout[G]{L}$ together with an edge $(u,v) \in \bndout[G]{L}$; since $u \in L$, we have $\{u^-\!,u^+\} \subseteq L'$. If $v^- \notin L'$, then $(u^+,v^-) \in \bndout[G']{L'}$, so $v^- \in \nbrout[G']{L'}$. Otherwise $v^- \in L'$, and since $v \notin L$ we must have $v^+ \notin L'$; then $(v^-,v^+) \in \bndout[G']{L'}$, so $v^+ \in \nbrout[G']{L'}$. In either case at least one of $v^-\!, v^+$ lies in $\nbrout[G']{L'}$, and these witnesses are distinct for different vertices $v \in \nbrout[G]{L}$. Hence,
    \[
        |\nbrout[G]{L}| \le |\nbrout[G']{L'}| \le |\bndout[G']{L'}| \le \cutsize.
    \]
    Next we bound the volume:
    \begin{align*}
        \volout[G]{L} = \sum_{v \in L} \dgout[G]{v} \overset{(1)}{=} \sum_{v \in L} \dgout[G']{v^+} \overset{(2)}{\le} \volout[G']{L'} - |L \setminus \{s\}|, \tag{$*$} \label{eq:star}
    \end{align*}
    where $(1)$ uses $\dgout[G']{v^+} = \dgout[G]{v}$, and $(2)$ holds because for each $v \in L$ we have $\{v^-\!, v^+\} \subseteq L'$ with $\dgout[G']{v^-} = 1$ for $v \ne s$, so $\{v^+ \mid v \in L\}$ and $\{v^- \mid v \in L \setminus \{s\}\}$ are disjoint subsets of $L'$. Since $\volout[G']{L'} \le \volest$, this shows $\volout[G]{L} \le \volest$.

    It remains to show that the cut is valid, i.e., that $V(G) \setminus (L \cup \nbrout[G]{L}) \ne \emptyset$, or equivalently $L \cup \nbrout[G]{L} \ne V(G)$.
    Suppose for contradiction that $L \cup \nbrout[G]{L} = V(G)$, and let $k = |\nbrout[G]{L}|$; since $L$ and $\nbrout[G]{L}$ are disjoint, $|L| = n - k$.
    We distinguish two cases. If $m' \ge 2nk$, then
    \[
        m = \volout[G]{L} + \volout[G]{\nbrout[G]{L}}.
    \]
    For the first term, \eqref{eq:star} gives
    \[
        \volout[G]{L} \le \volout[G']{L'} - |L \setminus \{s\}| \le \volest - (n-k-1) \le \frac{m'}{4} - n + k + 1.
    \]
    For the second term, since $G$ is simple, each of the $k$ vertices of $\nbrout[G]{L}$ has out-degree at most $n-1$, so
    \[
        \volout[G]{\nbrout[G]{L}} \le k (n - 1) = nk - k \le \frac{m'}{2} - k.
    \]
    Adding these yields
    \[
        m \le \frac{m'}{4} - n + k + 1 + \frac{m'}{2} - k < m' - n + 1 = m,
    \]
    a contradiction which shows that $m' < 2nk$. But then
    \[
        \volest \ge \volout[G]{L} \overset{(3)}{\ge} |L|\cdot k = (n - k) \cdot k \overset{(4)}{\ge} \frac{n}{2} \cdot k = \frac{nk}{2} \overset{(5)}{>} \frac{m'}{4} \ge \volest,
    \]
    where $(3)$ holds because every $v \in L$ has $\dgout[G]{v} > \cutsize \ge k$, $(4)$ uses $k \le \cutsize \le \frac{n}{2}$, and $(5)$ uses $m' < 2nk$. This is again a contradiction; therefore, $L \cup \nbrout[G]{L} \ne V(G)$.
\end{proof}

We now combine this reduction with our incremental edge-cut data structure (\cref{thm:local_ec}) to prove \cref{thm:local_vc}.
Besides a copy of $G$, our \textsc{IncrementalRootedVC} data structure will require a copy of the split graph $G'$ as doubly linked adjacency lists, in which \emph{every} vertex (including $s$) is split; that is, $V(G') = \{v^-\!, v^+ \mid v \in V(G)\}$ and $E(G') = \{(u^+\!, v^-) \mid (u,v) \in E(G)\} \cup \{(v^-\!, v^+) \mid v \in V(G)\}$.
This modification makes $G'$ independent of $s$, so a single copy of $G'$ can be reused for all invocations of \textsc{IncrementalRootedVC} in \cref{alg:main},
even if these instances use different source vertices $s$: during an invocation, some edges of $G'$ might get reversed, but \cref{alg:inc_local_ec} undoes all these reversals before returning, leaving $G'$ in its original state for the next invocation.

This copy of $G'$ can be maintained with constant overhead per edge insertion:
Whenever an edge $(u,v)$ gets inserted into $G$, we insert the edge $(u^+, v^-)$ into $G'$.
Finally, we note that, although $s$ is split in $G'$, the adjacency list of $G'$
encodes all the information necessary for us to run \cref{alg:inc_local_ec} on the split graph where the vertex $s$ is not split.
Indeed, the out-neighbors of the unsplit vertex $s$ are precisely the out-neighbors of the split vertex $s^+$, and we can interpret any edge leading to the split vertex $s^-$ as an edge leading to $s$ instead.

\theoLocalVC*
\begin{proof}
    Our data structure is given $G$ together with the split graph $G'$ in which every vertex, including $s$, is split, as described above. For the analysis we instead regard $G'$ as the split graph of $G$ with $s$ left unsplit, exactly as in \cref{lem:vc_to_ec,lem:ec_to_vc}; then $G'$ has $n' = 2n-1$ vertices and $m' = m+n-1$ edges.
    We defer the explanation of how the algorithm operates on the fully split graph it actually receives to the implementation details below.

    Our data structure runs \textsc{IncrementalRootedEC} (\cref{thm:local_ec}) on $G'$ with parameters $\cutsize' = \cutsize$, $\volest' = 2\volest$, and source $s$. Each edge $(u,v)$ inserted into $G$ is passed on as the edge $(u^+, v^-)$ inserted into $G'$. Whenever \textsc{IncrementalRootedEC} reports an edge cut $(L', V(G') \setminus L')$, we implicitly form $L = \{v \in V(G) \mid \{v^-\!, v^+\} \subseteq L'\}$ and output the separator $S = \nbrout[G]{L}$; whenever it reports $\bot$, we output $\bot$. In the implementation details below, we explain how to compute $S$ efficiently.

    \subparagraph{Correctness.}
    We first check that \textsc{IncrementalRootedEC} is invoked with valid parameters: $\cutsize'$ and $\volest'$ are positive, and \cref{thm:local_ec} requires $\volest' \le \frac{m'}{2\cutsize}$, which holds since $\volest' = 2\volest \le \frac{m}{12\cutsize} \le \frac{m'}{2\cutsize}$.

    Suppose \textsc{IncrementalRootedEC} reports a cut $(L', V(G') \setminus L')$. By \cref{thm:local_ec} this is a \voledgecut{\cutsize}{\le}{3\volest'\cutsize}{s} in $G'$, that is, a $(\cutsize, \le\!6\volest\cutsize)$-edge cut. Since $6\volest\cutsize \le 6 \cdot \frac{m}{24\cutsize} \cdot \cutsize = \frac{m}{4} \le \frac{m'}{4}$, and since $\cutsize \le \frac{n}{2}$ and $\mindegout[G] > \cutsize$, \cref{lem:ec_to_vc} applies and shows that $(L, \nbrout[G]{L}, V(G) \setminus (L \cup \nbrout[G]{L}))$ is an \volvertcut{\cutsize}{\le}{6\volest\cutsize}{s} in $G$.
    Hence the vertex cut $S = \nbrout[G]{L}$ that we output belongs to a valid \volvertcut{\cutsize}{\le}{6\volest\cutsize}{s}.

    It remains to show that we do not output $\bot$ prematurely. For this, let $(L, S, R)$ be the particular \volvertcut{\cutsize}{\le}{\volest}{s} from the statement of the theorem. %Suppose that $G$ has a \volvertcut{\cutsize}{\le}{\volest}{s} $(L, S, R)$.
    By \cref{lem:vc_to_ec}, the pair $(L', V(G') \setminus L')$ with $L' = \{v^-\!,v^+ \mid v \in L\} \cup \{v^- \mid v \in S\}$ is a \voledgecut{\cutsize}{\le}{2\volest}{s} in $G'$; applying the lemma to the current graph after each insertion, it remains valid for as long as $(L,S,R)$ remains a valid \volvertcut{\cutsize}{\le}{\volest}{s}.
    Since each insertion into $G$ induces exactly one insertion into $G'$, a randomness-adaptive adversary against our data structure induces a randomness-adaptive adversary against \textsc{IncrementalRootedEC}.
    Thus, by \cref{thm:local_ec} (with $\volest' = 2\volest$), with probability at least $\frac{1}{2}$, \textsc{IncrementalRootedEC} outputs $\bot$ only once $(L', V(G') \setminus L')$ is no longer valid; in that case our data structure outputs $\bot$ only once $(L,S,R)$ is no longer valid.

    \subparagraph{Implementation details.}
    The algorithm is actually given the graph $G'$ in which \emph{every} vertex, including $s$, is split. To make it behave as on the single-copy split graph above, we run \textsc{IncrementalRootedEC} from the source $s^+$ and, at the start of each of its outer iterations, additionally mark $s^-$ as reached, adding it to $L'$.
    The only edge out of $s^-$ is the internal edge $(s^-, s^+)$, whose endpoints both lie in $L'$, so this reproduces the reachability from an unsplit $s$, at $\cO(1)$ extra cost per iteration.

    To maintain the vertex cut $S$ efficiently, we first make a few observations.
    Recall that $L = \{v \in V(G) \mid \{v^-\!, v^+\} \subseteq L'\}$ and $S = \nbrout[G]{L}$.
    A vertex $v$ enters $L$ exactly when both $v^-$ and $v^+$ have been added to $L'$.
    When this happens, $S$ must be updated in two ways:
    since $v$ is now in $L$, it leaves $\nbrout[G]{L}$ and must be removed from $S$ if present; and every out-neighbor of $v$ in $G$ that is not already in $L$ must be added to $S$.

    We maintain $L$ implicitly: for each vertex $v \in V(G)$ we record whether $v^- \in L'$ and whether $v^+ \in L'$.
    This lets us test whether $v \in L$ in $\cO(1)$ time.
    We update this information in $\cO(1)$ time whenever \textsc{IncrementalRootedEC} adds a vertex to $L'$. 
    We maintain $S$ explicitly as a doubly linked list together with, for each vertex $u \in V(G)$, a pointer to its position in $S$ (or $\bot$ if $u \notin S$). This way, a vertex can be added to or removed from $S$ in $O(1)$ time.
    Whenever \textsc{IncrementalRootedEC} adds the vertex $v^+$ or $v^-$ to $L'$,
    we test whether $v$ should be added to $L$,
    and if so, update $S$ as described above.
    Maintaining $S$ this way takes $\cO(\volout[G]{L})$ time per outer iteration of the for loop in \cref{alg:inc_local_ec}.

    There is one technical complication here.
    The bound $\volout[G]{L} \le \volout[G']{L'} \le 6\volest\cutsize$ from \cref{lem:getcut_correct} and \cref{lem:ec_to_vc} only holds when \textsc{IncrementalRootedEC} \emph{outputs} the set $L'$; while $L'$ is still growing, $\volout[G]{L}$ can be far larger, so updating $S$ immediately after every change to $L$ could take $\gg \volest \cutsize$ time.
    To avoid this, we record all changes to $L$ in a linked list and apply them to $S$ only once \textsc{IncrementalRootedEC} outputs $L'$, at which point $\volout[G]{L} \le 6\volest\cutsize$.
    If \cref{alg:inc_local_ec} instead exits the while loop and resets $L$ to $\{s\}$, we discard this list and reset $L$ and $S$.

    \subparagraph{Running time.}
    By \cref{thm:local_ec}, \textsc{IncrementalRootedEC} on $G'$ with parameters $\cutsize$ and $2\volest$ processes $\ell$ insertions in $\cO(2\volest\cutsize^2 + \ell) = \cO(\volest\cutsize^2 + \ell)$ time.
    Maintaining $L$ adds a constant-factor overhead compared to how \textsc{IncrementalRootedEC} maintains $L'$.
    Maintaining $S$ costs $\cO(\volest\cutsize)$ per iteration of the for loop in \cref{alg:inc_local_ec}; over the at most $\cutsize + 1$ such iterations, this sums to $\cO(\volest\cutsize^2)$.
    Since $G'$ is given as input and restored after use, no time is spent constructing it.
    In total, the data structure runs in $\cO(\volest\cutsize^2 + \ell)$ time.
\end{proof}

\section{Auxiliary data structures} \label{sec:aux_datastructs}

This section covers the existing data structures for incremental reachability and incremental max-flow from the literature that we use in \cref{alg:main}.
For completeness, we describe how each of these structures works internally.
We start by presenting a folklore data structure for maintaining single-source reachability.

\begin{theorem}[{\cite{italiano1986}}] \label{clm:inc_reachability}
    Given a graph $G$ with $n$ vertices and a source vertex $u \in V(G)$, there is a data structure that maintains $G$ under edge insertions and, after every insertion, reports whether every vertex $v \in V(G)$ is reachable from $u$ in $G$.
    The total running time is $\cO(m)$, where $m$ is the number of edges in $G$ after all edge insertions.
\end{theorem}
\begin{proof}
    The data structure maintains a set $T$ of vertices reachable from $u$.
    Upon insertion of an edge $(v, w)$, we check whether $v \in T$.
    If so, and if $w \notin T$, we add $w$ to $T$ and recurse on the outgoing edges of $w$.
    This process takes $\cO(m)$ time in total over all edge insertions, as every edge is processed at most twice: when it is inserted and when its tail is added to $T$.
    Maintaining a counter for $|T|$, we can then answer whether all vertices are reachable from $u$ in $\cO(1)$ time by checking whether $|T| = n$.
\end{proof}

Next, we will use this data structure to maintain whether a given graph undergoing edge insertions is strongly connected.
This is done by maintaining two copies of the data structure, on the graph and the reversed graph, respectively.

\lemmIncConn*
\begin{proof}
    We fix an arbitrary vertex $u \in V(G)$ and maintain the single-source reachability from $u$ in both the graph $G$ and the reversed graph $\reversed{G}$ with the incremental data structure from \cref{clm:inc_reachability}, where the total running time is $\cO(m)$.
    We show that $G$ is strongly connected if and only if every vertex $v \in V(G)$ is reachable from $u$ both in $G$ and in $\reversed{G}$.

    If there is a vertex $v$ that is not reachable from $u$ in $G$, then $G$ is trivially not strongly connected.
    Similarly, if there is a vertex $v$ that is not reachable from $u$ in $\reversed{G}$, then $u$ is not reachable from $v$ in $G$, so $G$ is also not strongly connected.

    Conversely, assume that every vertex is reachable from $u$ in both $G$ and $\reversed{G}$.
    Now fix any two vertices $v, w \in V(G)$.
    Since $v$ is reachable from $u$ in $\reversed{G}$, the vertex $u$ is reachable from $v$ in $G$; and by assumption $w$ is reachable from $u$ in $G$.
    Concatenating these two paths shows that $w$ is reachable from $v$. As $v$ and $w$ were chosen arbitrarily, $G$ is strongly connected.
\end{proof}

Finally, we mention the data structure we use to maintain exact $(s, t)$-vertex cuts under edge insertions, using max-flow.
The max-flow data structure of \cite{GuptaK21} internally uses the reachability structure introduced above to maintain a reachability tree in the residual graph, updating the flow and recomputing the tree whenever the flow value increases.
Since we are only interested in vertex cuts of size $\cutsize$, we can terminate this process once the flow becomes too big.

\thmMaxFlowVC*
\begin{proof}
    The data structure maintains a maximum $(s, t)$-flow in a modified graph $G'$ on $2 n$ vertices and $m + 2n$ edges using the incremental max-flow data structure of \cite{GuptaK21}.
    The graph $G'$ is obtained from $G$ in the standard way by splitting each vertex $v \in V(G)$ into two vertices $v^-\!, v^+$, adding an edge $(v^-\!, v^+)$ of capacity $1$, and an edge $(v^+\!, v^-)$ of infinite capacity;
    and replacing each edge $(u, v) \in E(G)$ with an edge $(u^+\!, v^-)$ of infinite capacity.
    Then, a minimum $(s, t)$-vertex cut in $G$ corresponds to a minimum $(s^+\!, t^-)$-edge cut in $G'$.
    
    More precisely, if $(L', R')$ is a minimum $(s^+\!, t^-)$-edge cut in $G'$,
    then $S$ is a minimum $(s, t)$-vertex cut where $S = \set{v \in V(G) \mid v^- \in L' \land v^+ \in R'}$.
    The data structure of \cite{GuptaK21} maintains not only the current flow value, but also the set of vertices $L'$ of some minimum $(s^+\!, t^-)$-edge cut $(L', R')$ in $G'$.
    Thus, the structure can also maintain the set $S$ of the corresponding $(s, t)$-vertex cut in $G$.

    The total running time of the data structure of \cite{GuptaK21} for a sequence of $\ell$ edge insertions is $O(f m + \ell)$
    where $f$ is the final flow value. In our case, we will terminate the augmentation process once the flow value exceeds $\cutsize$ and report that no $(s, t)$-vertex cut of size $\cutsize$ exists, and so the runtime will be $O(\cutsize m + \ell)$.
    We also need an additional $O(\cutsize \ell)$ total time to output the set $S$ after every edge insertion.
    Thus, the total running time is $O(\cutsize (m + \ell))$.
\end{proof}

\section{Declaration of AI usage}
All lemmas, theorems, algorithms, and results were obtained without the use of artificial intelligence.
We used Gemini 3.1 Pro for creating and modifying the tikz source code of \cref{fig:lower_bound_transformation}.
In the final writing stages, we used large language models for spell-checking. 

\printbibliography

\end{document}